%% file: main_511.tex
\documentclass[12pt, journal, onecolumn,draftclsnofoot]{IEEEtran}
\usepackage[margin=1in]{geometry}
\usepackage[utf8]{inputenc}
\usepackage{amsmath}
\usepackage{amsfonts,color}
\usepackage{amsthm}
\usepackage{amssymb}
\usepackage{graphicx,subcaption}
\usepackage{bbm}
\usepackage{tikz}

\input{macro_wnchen.tex}

\newtheorem{theorem}{Theorem}

\newtheorem{prop}{Proposition}
\newtheorem{cor}{Corollary}
\newtheorem{lem}{Lemma}
\newtheorem{rem}{Remark}

\newcommand{\citep}{\cite}

\title{Fisher information under local differential privacy}

\author{Leighton Pate Barnes, Wei-Ning Chen, and Ayfer \"Ozg\"ur \\
Stanford University, Stanford, CA 94305\\
 Email: \{lpb, wnchen, aozgur\}@stanford.edu} %\thanks{A conference version of this paper has been submitted to COLT 2020 and is under review. It does not include new contributions such as the blackboard model, high-privacy regime for the sparse Bernoulli model, and detailed regularity conditions.}}

\begin{document}

\maketitle

\begin{abstract}
We develop data processing inequalities that describe how Fisher information from statistical samples can scale with the privacy parameter $\varepsilon$ under local differential privacy constraints. These bounds are valid under general conditions on the distribution of the score of the statistical model, and they elucidate under which conditions the dependence on $\varepsilon$ is linear, quadratic, or exponential. We show how these inequalities imply order optimal lower bounds for private estimation for both the Gaussian location model and discrete distribution estimation for all levels of privacy $\varepsilon>0$. We further apply these inequalities to sparse Bernoulli models and demonstrate privacy mechanisms and estimators with order-matching squared $\ell^2$ error.
\end{abstract}

%\begin{keywords}
%Fisher information, local differential privacy, data processing, statistical estimation
%\end{keywords}

\section{Introduction}
In the model of local differential privacy \cite{warner,dwork1,dwork2,dwork3}, sensitive data is released to an aggregator or centralized processor only after having been processed by a privatization mechanism. This privatization mechanism distorts the data in such a way that it is statistically guaranteed to not reveal too much about the underlying sensitive data. There is an inherent trade-off between the degree to which the data is distorted by the mechanism (and therefore the amount of privacy achieved), and the utility of the data for performing statistical inference and estimation tasks.

One measure of the information conveyed by a statistical sample for estimating a parameter is the so-called Fisher information, which describes how a family of probability distributions changes as one varies the parameter of interest. Under some mild regularity conditions, the Fisher information at a point $\theta$ in the space of possible parameters immediately gives a lower bound on the squared $\ell^2$ risk for estimating $\theta$ via the well-known Cram\'er-Rao bound for unbiased estimators \cite{Cramer,Rao,Cover--Thomas2006,Tsybakov2008}. More generally, Fisher information describes the complexity of estimation problems in an asymptotic sense locally around $\theta$ \cite{Hajek1972local,Vandervaart2000}; and a Bayesian version of the Cram\'er-Rao bound known as the van Trees inequality can be used to give lower bounds that hold for any estimator (including arbitrarily biased estimators) \cite{gill}.

In this paper, we consider the problem of estimating a parameter $\theta\in\mathbb{R}^d$ from $n$ independent statistical samples that have been processed by an $\varepsilon$-locally differentially private mechanism. See Figure \ref{fig1}.
%This task is sufficiently general to encompass many problems of interest such as estimating the mean of a Gaussian or other location model, estimating a discrete distribution, or estimating a smooth density via parametric reduction.
We characterize the Fisher information from the privatized samples, and provide strong data processing inequalities that describe how the Fisher information can scale with the privatization parameter $\varepsilon$. These data processing inequalities are valid under very general conditions on the tail of the score function random variable, and elucidate under which conditions the dependence on $\varepsilon$ is linear, quadratic, or exponential. Using the van Trees inequality, we recover in a unified way order-wise optimal lower bounds on the minimax squared $\ell^2$ risk for Gaussian mean estimation and discrete distribution estimation at all levels of privacy $\varepsilon>0$, matching the lower bounds from \cite{duchi2013local,yebarg2018,duchi2019lower} with simpler and more transparent proofs. Our results also apply to a sequential interaction model for local differential privacy via a straightforward consequence of the chain-rule for Fisher information; and they can even be applied to a fully interactive blackboard model by characterizing the Fisher information from the entire blackboard transcript, therefore extending, for example, earlier bounds in \cite{yebarg2018} to fully interactive models.

We further demonstrate the utility of this framework by developing lower bounds for other statistical models such as a sparse Bernoulli model with $X_i \sim \prod_{j=1}^d\mathsf{Bern}(\theta_j)$ and $\sum_{j=1}^d \theta_j \leq s$, and demonstrate a privatization mechanism with matching error. This model is interesting in that  when $s=1$, it provides an example where the dependence of the minimax squared $\ell^2$ risk on $\varepsilon$ is exponential even if the $d$ components of each sample $X_i$ are independent of each other. This is in contrast to mean estimation for the Gaussian and the dense Bernoulli model (when $s$ is of order $d$), where the dependence on $\varepsilon$ is linear \cite{duchi2019lower}. When $s>1$, the sample-size penalty due to privatization is of the order $\frac{s\log d}{\varepsilon}$ in the privacy regime $\log d \preceq \varepsilon \preceq s\log d$. This penalty scales linearly only with the sparsity $s$ rather than the ambient dimension $d$, which opens up the possibility of private estimation with a more modest penalty provided that the data can be assumed to be sparse in a certain sense.

Our Fisher information approach to lower bounds under privacy constraints is motivated by recent results for statistical estimation under communication constraints such as \cite{allerton,isit,barnes}. In both the privacy constrained and communication constrained cases, the tail behavior of the score function plays a central role in determining how the risk can scale. Other works such as \cite{archayaetal2,archayaetal,duchi2019lower} have also noted the connection between communication and privacy constraints in statistical estimation. In contrast with \cite{ruan}, we analyze the Fisher information from the induced distribution of the privatized samples $Y_1,\ldots, Y_n$, rather than that from the original statistical model of the $X_i$'s. In \cite{ruan}, Ruan and Duchi observe that the latter has limited applicability for capturing the local complexity of private estimation problems (see also \cite{rohde}), while our paper shows that the former is a powerful measure for the same. Indeed, from the local asymptotic minimax point of view, it is natural to expect the Fisher information from the privatized samples to play a role in the complexity of the private estimation problem, but until now it remained unclear how to characterize or bound this Fisher information for any privatization mechanism satisfying the local differential privacy condition.

The main contributions of our paper can be summarized as follows:
\begin{itemize}
\item We introduce a framework for characterizing Fisher information from $\varepsilon$-differentially privatized samples. Under very general conditions on the statistical model, we provide upper bounds on the Fisher information  that show that the dependence on $\varepsilon$ is dictated by  the tail of the score function random variable. These bounds continue to hold even when samples are released in an interactive fashion through a shared blackboard.  Even though statistical estimation under privacy constraints has been of significant recent interest, to the best of our knowledge there are no known bounds on  the Fisher information from privatized samples.
\item We show that the bounds on Fisher information easily lend themselves to order optimal lower bounds on the minimax squared $\ell^2$ risk of statistical estimation under privacy constraints. In particular, we recover in a unified way lower bounds developed separately for different statistical models in the literature, such as Gaussian mean estimation \cite{duchi2019lower} and discrete distribution estimation \cite{yebarg2018}, in the latter case extending the bounds to fully interactive models.
\item To demonstrate the generality of our approach, we apply our bounds to a sparse
Bernoulli model, for which we also develop optimal privacy mechanisms. We show that our framework can be flexibly applied to different parameter regimes of this model and the dependence of the minimax risk on $\varepsilon$ can be exponential or linear depending on the parameter regime of interest.
\end{itemize}

\begin{figure}
\begin{center}
\begin{tikzpicture}
\node at (2, 1.4) {$P_\theta$}; 
\draw [->] (1.8, 1.2) -- (0.2, 0.2); 
\draw [->] (1.9, 1.2) -- (1.1, 0.2); 
\draw [->] (2.1, 1.2) -- (2.9, 0.2); 
\draw [->] (2.2, 1.2) -- (3.8, 0.2); 
\draw (0,0) circle (0.2cm); \node [below] at (0,-0.2) {$X_1$}; 
\draw (1,0) circle (0.2cm); \node [below] at (1,-0.2) {$X_2$}; 
\node at (2,0) {$\cdots$}; 
\draw (3,0) circle (0.2cm); \node [below] at (3,-0.2) {$X_{n-1}$}; 
\draw (4,0) circle (0.2cm); \node [below] at (4,-0.2) {$X_{n}$}; 
\draw [->] (0, -0.7) -- (0, -1.4); \draw [->] (1, -0.7) -- (1, -1.4); 
\draw [->] (3, -0.7) -- (3, -1.4); \draw [->] (4, -0.7) -- (4, -1.4);
\draw (-0.4,-2.5) rectangle (.4,-1.5); \node at (0,-2) {$Q$}; 
\draw (0.6,-2.5) rectangle (1.4,-1.5); \node at (1,-2) {$Q$}; 
\node at (2,-2) {$\cdots$}; 
\draw (2.6,-2.5) rectangle (3.4,-1.5); \node at (3,-2) {$Q$}; 
\draw (3.6,-2.5) rectangle (4.4,-1.5); \node at (4,-2) {$Q$}; 
\draw [->] (0, -2.6) -- (0, -3.2); \node [below] at (0,-3.2) {$Y_1$}; \draw [->] (0, -3.8) -- (0, -4.4);
\draw [->] (1, -2.6) -- (1, -3.2); \node [below] at (1,-3.2) {$Y_2$}; \draw [->] (1, -3.8) -- (1, -4.4);
\draw [->] (3, -2.6) -- (3, -3.2); \node [below] at (3,-3.2) {$Y_{n-1}$}; \draw [->] (3, -3.8) -- (3, -4.4);
\draw [->] (4, -2.6) -- (4, -3.2); \node [below] at (4,-3.2) {$Y_{n}$}; \draw [->] (4, -3.8) -- (4, -4.4);
\draw (-0.3,-5.5) rectangle (4.3,-4.5);  % \node at (2,-3) {transcript $Y$}; 
\node at (2,-5) {centralized processor}; 
\draw [->] (2,-5.5) -- (2, -6); \node [below] at (2,-6) {$\hat{\theta}$}; 
%\node [left] at (0, -1) {$k$ bits}; \node [right] at (4,-1) {$k$ bits}; 
\end{tikzpicture}
\caption{An estimation system where sensitive data $X_1,\ldots X_n$ is processed by the privatization mechanism $Q(y|x)$ before being released to the centralized processor that will use the data for statistical inference tasks such as estimating the parameter $\theta$.}
\label{fig1}
\end{center}
\end{figure}
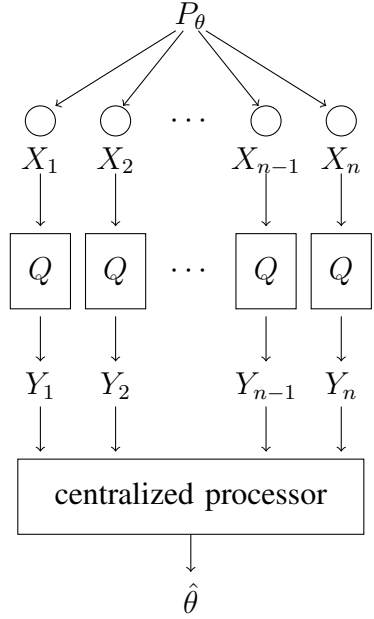 

\subsection{Preliminaries}
Let $(\mathcal{X},\mathcal{A})$ and $(\mathcal{Y},\mathcal{B})$ be measurable spaces and suppose that $\{P_\theta\}_{\theta\in\Theta}$ for $\Theta \subseteq \mathbb{R}^d$ is a family of probability measures on $(\mathcal{X},\mathcal{A})$ that is dominated by some sigma-finite measure $\mu$. Denote the density of $P_\theta$ with respect to $\mu$ by $f(x|\theta)$. Let
$$X_1,\ldots,X_n \overset{\text{i.i.d.}}{\sim} P_\theta$$
with $\theta$ being the parameter of interest that we are trying to estimate. We say that a regular conditional distribution $$Q:\mathcal{B}\times\mathcal{X}\to [0,1]$$
is an $\varepsilon$-differentially private mechanism if
\begin{equation}\frac{Q(S|x)}{Q(S|x')} \leq e^\varepsilon \label{eq:epsDP} \end{equation}
for any $x,x'\in\mathcal{X}$ and $S\in\mathcal{B}$.

Suppose that conditioned on $X_i=x_i$, we draw $Y_i$ independently from $Q(\cdot|x_i)$ where $Q$ satisfies the $\varepsilon$ differentially priviate condition \eqref{eq:epsDP} with $\varepsilon>0$.  The privatization mechanisms can either be the same for each sample, or they can vary across different samples as long as each mechanism satisfies the condition \eqref{eq:epsDP}. We will generally assume for simplicity that each mechanism is the same. In this case the $Y_i$ have the marginal probability distribution
$$Q_\theta(S) = \int Q(S|x)f(x|\theta)d\mu(x) \; .$$
By \eqref{eq:epsDP}, if $Q(S|x)=0$ for some $x\in\mathcal{X}$ then $Q(S|x') = 0$ for all $x'\in\mathcal{X}$. We can therefore assume that both $\{Q(\cdot|x)\}_{x\in\mathcal{X}}$ and $\{Q_\theta\}_{\theta\in\Theta}$ form dominated families with $Q(\cdot|x)<<\nu$ for all $x\in\mathcal{X}$ and $Q_\theta<<\nu$ for all $\theta\in\Theta$ for some sigma-finite measure $\nu$. Abusing notation slightly let $f(y|\theta)$ be the density of $Q_\theta$ with respect to $\nu$, and let $Q(y|x)$ denote the density of $Q(\cdot|x)$ with respect to $\nu$.

Instead of working with measures, we will find it more convenient to work with the corresponding densities of the probability distributions. The following proposition shows that the condition \eqref{eq:epsDP} implies a similar condition for the density $Q(y|x)$, and this is the form that will be most useful in the subsequent sections.

\begin{prop} \label{prop1}
For any $x,x'\in\mathcal{X}$ and $\nu$-almost any $y\in\mathcal{Y}, \; \frac{Q(y|x)}{Q(y|x')} \leq e^\varepsilon \; .$
\end{prop}
\noindent A proof of Proposition \ref{prop1} is included in Appendix \ref{app1}.

\section{Upper Bounds on Fisher Information}
In this section, we introduce the relevant Fisher information quantities and then show how the Fisher information from the privatized samples $Y_i$ can be upper bounded bounded in terms of the local differential privacy parameter $\varepsilon$. Recall that in the context of Fisher information, the score function associated with the statistical model $P_\theta$ is defined by
\begin{align*}
S_\theta(x) & = \nabla_\theta \log f(x|\theta) \\
& = \left(\frac{\partial}{\partial\theta_1} \log f(x|\theta),\ldots,\frac{\partial}{\partial\theta_d} \log f(x|\theta)\right)^T \; ,
\end{align*}
and the Fisher information matrix for estimating $\theta$ from a sample $Y$ is
$$I_Y(\theta) = \mathbb{E}\left[\left(\nabla_\theta\log f(Y|\theta)\right)\left( \nabla_\theta\log f(Y|\theta)\right)^T\right]$$
where the expectation is understood to be taken with respect to the ``true'' distribution with parameter $\theta$. All logs are taken with respect to the natural base. In order to ensure that these quantities are well-defined, and that we can apply the van Trees inequality below, we require certain regularity conditions on the statistical model $P_\theta$. In particular we assume that the square-root densities $\sqrt{f(x|\theta)}$ are continuously differentiable with respect to each $\theta_j$ and that the Fisher information from $X$ exists and is finite. For more on these conditions and how they are used see Appendix \ref{app:reg}.

We are interested in this Fisher information quantity, in part, because it can provide bounds on the risk in estimating $\theta$ from the samples $Y_1,\ldots,Y_n$. Fisher information is by definition a local quantity that is defined at each $\theta\in\Theta$ and describes the local complexity of estimating that particular $\theta$ value asymptotically as the number of samples $n$ increases. More concretely, if the statistical model $Q_\theta$ is \emph{differentiable in quadratic mean} \footnote{is implied by the assumptions we have made above without any additional assumptions on $Q(y|x)$ other than it being an $\varepsilon$ differentially private mechanism} \cite{Vandervaart2000}, meaning that $\sqrt{f(y|\theta)}$ has a derivative in a certain $L^2$ sense, then the local asymptotic risk around $\theta$ is lower bounded as follows:
\begin{align*}
\sup_A\liminf_{n\to\infty}\sup_{h\in A}\mathbb{E}_{\theta+h/\sqrt{n}}\left\|\sqrt{n}\left(\hat\theta_n(Y_1,\ldots,Y_n) - \left(\theta+\frac{h}{\sqrt{n}}\right)\right)\right\|_2^2 & \geq \mathsf{Tr}(I_Y(\theta)^{-1}) \\
& \geq \frac{d^2}{\mathsf{Tr}(I_Y(\theta))}
\end{align*}
where $A$ is any finite subset of $\mathbb{R}^d$ (e.g. see Theorem 8.11 in \cite{Vandervaart2000}). Similarly, under these conditions, for each $\theta$, there exists a sequence of estimators $\hat\theta_n(Y_1,\ldots,Y_n)$ such that
$$ \sup_A\limsup_{n\to\infty}\sup_{h\in I}\mathbb{E}_{\theta+h/\sqrt{n}}\left\|\sqrt{n}\left(\hat\theta_n(Y_1,\ldots,Y_n) - \left(\theta+\frac{h}{\sqrt{n}}\right)\right)\right\|_2^2  \leq \mathsf{Tr}(I_Y(\theta)^{-1})$$
(e.g. see Theorem 8.14 in \cite{Vandervaart2000}). In this way the Fisher information can determine both upper and lower bounds and is of fundamental importance in the parameter estimation problem. In the sequel, we develop upper bounds on $\mathsf{Tr}(I_Y(\theta))$, which immediately lead to lower bounds on local asymptotic risk. Additionally, by upper bounding Fisher information uniformly across all $\theta\in\Theta$ (or a subset of $\Theta$), and using a Bayesian Cram\'er-Rao bound, we are able to get more global minimax lower bounds. For this we'll use a multivariate version of the van Trees inequality due to Gill and Levit \cite{gill}, which bounds the average $\ell^2$ risk by
\begin{align} \label{eq:vt}
\int_\Theta \mathbb{E}\|\hat\theta-\theta\|_2^2\lambda(\theta)d\theta \geq \frac{d^2}{\int_\Theta\text{Tr}(I_{Y_1,\ldots,Y_n}(\theta))\lambda(\theta)d\theta + J(\lambda)}
\end{align}
where $\lambda(\theta) = \prod_{j=1}^d\lambda_j(\theta_j)$ is a prior for the parameter $\theta$ and $J(\lambda)$ is the Fisher information associated with the prior $\lambda$:
$$J(\lambda) = \sum_{j=1}^d \int \frac{\lambda_j'(\theta_j)^2}{\lambda_j(\theta_j)}d\theta_j  \; .$$
Assuming that $\Theta = [-B,B]^d$, the prior $\lambda$ can be chosen to minimize $J(\lambda)$ \cite{Tsybakov2008,borovkov}. This observation along with the independence of the $Y_i$, and upper bounding the average risk by the maximum risk, leads to
\begin{align} \label{eq:vt2}
\sup_{\theta\in\Theta} \mathbb{E}_\theta\|\hat\theta-\theta\|_2^2 \geq \frac{d^2}{n\sup_{\theta\in\Theta}\mathsf{Tr}(I_{Y}(\theta)) + \frac{d\pi^2}{B^2}} \; .
\end{align}
 We are therefore interested in upper bounding $\mathsf{Tr}(I_Y(\theta))$. To this end, we will need the following lemma:
\begin{lem}[Barnes et. al 2018 \cite{allerton}] \label{lem1}
The trace of the Fisher information matrix $I_Y(\theta)$ can be written as
$$\mathsf{Tr}(I_Y(\theta)) = \mathbb{E}_Y\|\mathbb{E}_X[S_\theta(X)|Y]\|_2^2 \; .$$
\end{lem}
For completeness a proof of Lemma \ref{lem1} is included in Appendix \ref{app2}. Using the characterization of Fisher information from Lemma \ref{lem1}, the following Propositions \ref{prop2}-\ref{prop5} show how in the differentially private setting, $\mathsf{Tr}(I_Y(\theta))$ can be upper bounded under various assumptions on the tail of the score function random vector $S_\theta(X)$. We see that depending on the tail behavior, there can be a qualitatively different upper bound in terms of $\varepsilon$, i.e., it can be linear, quadratic, or exponential in $\varepsilon$. In Figure \ref{fig2} we summarize these conditions and the corresponding upper bounds.
\begin{prop} \label{prop2}
If $\mathbb{E}[\langle u,S_\theta(X)\rangle^2] \leq I_0$ for any unit vector $u\in\mathbb{R}^d$, then $$\mathsf{Tr}(I_Y(\theta)) \leq I_o (e^\varepsilon-1)^2 \; .$$
\end{prop}
\begin{proof}
Using Lemma \ref{lem1},
\begin{equation} \label{lems}
\mathsf{Tr}(I_Y(\theta)) = \mathbb{E}_Y\|\mathbb{E}_X[S_\theta(X)|Y]\|_2^2 \; .
\end{equation}
For a fixed $y$ let $$u = \frac{\mathbb{E}[S_\theta(X)|Y=y]}{\|\mathbb{E}[S_\theta(X)|Y=y]\|_2}$$
so that
\begin{align}
\|\mathbb{E}[S_\theta(X)|Y=y]\|_2 & = \langle u, \mathbb{E}[S_\theta(X)|Y=y] \rangle \nonumber \\
& = \mathbb{E}[\langle u, S_\theta(X) \rangle |Y=y] \nonumber \\
& = \frac{1}{f(y|\theta)}\mathbb{E}[\langle u, S_\theta(X) \rangle Q(y|X)] \; . \label{eq1}
\end{align}
Let $c_\text{min}(y) = \min_x Q(y|x)$ and $c_\text{max}(y) = \max_x Q(y|x).$ We can assume that $c_\text{min}(y) > 0$. Following \eqref{eq1},
\begin{align} \label{eq2}
\frac{1}{f(y|\theta)}\mathbb{E}[\langle u, S_\theta(X) \rangle Q(y|X)] & \leq \frac{1}{c_\text{min}(y)}\bigg(\int_{\{x:\langle u,S_\theta(x) \rangle \geq 0\}} \langle u,S_\theta(x) \rangle Q(y|x) f(x|\theta) d\mu(x) \nonumber\\
 & + \int_{\{x:\langle u,S_\theta(x) \rangle < 0\}} \langle u,S_\theta(x) \rangle Q(y|x) f(x|\theta) d\mu(x) \bigg) \nonumber\\
 & \leq \frac{1}{c_\text{min}(y)}\bigg(c_\text{max}(y)\int_{\{x:\langle u,S_\theta(x) \rangle \geq 0\}} \langle u,S_\theta(x) \rangle f(x|\theta) d\mu(x) \nonumber\\
 & + c_\text{min}(y)\int_{\{x:\langle u,S_\theta(x) \rangle < 0\}} \langle u,S_\theta(x) \rangle f(x|\theta) d\mu(x) \bigg) \; .
\end{align}
Note that score functions are mean zero and thus
\begin{equation} \label{eq3}
\int_{\{x:\langle u,S_\theta(x) \rangle \geq 0\}} \langle u,S_\theta(x) \rangle  f(x|\theta) d\mu(x) + \int_{\{x:\langle u,S_\theta(x) \rangle < 0\}} \langle u,S_\theta(x) \rangle f(x|\theta) d\mu(x) = 0 \; .
\end{equation}
Putting \eqref{eq2} together with \eqref{eq3},
\begin{equation*}
\frac{1}{f(y|\theta)}\mathbb{E}[\langle u, S_\theta(X) \rangle Q(y|X)] \leq (e^\varepsilon - 1)\int_{\{x:\langle u,S_\theta(x) \rangle \geq 0\}} \langle u,S_\theta(x) \rangle  f(x|\theta) d\mu(x) \; ,
\end{equation*}
and then squaring both sides and using Jensen's inequality yields
\begin{equation} \label{eq4}
\left(\frac{1}{f(y|\theta)}\mathbb{E}[\langle u, S_\theta(X) \rangle Q(y|X)]\right)^2 \leq (e^\varepsilon - 1)^2I_0 \; .
\end{equation}
The proposition is proved by combining \eqref{lems}, \eqref{eq1}, and \eqref{eq4}.
\end{proof}
\begin{rem}
If $0<\varepsilon<1$ then $(e^\varepsilon-1)^2 \leq (e-1)^2\varepsilon^2$ and Proposition \ref{prop2} implies
$$\mathsf{Tr}(I_Y(\theta)) \leq (e-1)^2 I_o \varepsilon^2 \; . $$
\end{rem}
\begin{prop} \label{prop3}
If $\mathbb{E}[\langle u,S_\theta(X)\rangle^2] \leq I_0$ for any unit vector $u\in\mathbb{R}^d$ then $$\mathsf{Tr}(I_Y(\theta)) \leq I_o e^\varepsilon \; .$$
\end{prop}
\begin{proof}
Following \eqref{eq1} from the proof above, Jensen's inequality implies
\begin{align*}
\mathbb{E}\left[\langle u,S_\theta(X)\rangle \frac{Q(y|X)}{\mathbb{E}[Q(y|X)]}\right]^2 & \leq \mathbb{E}\left[\langle u,S_\theta(X)\rangle^2\frac{Q(y|X)}{\mathbb{E}[Q(y|X)]}\right] \\
& \leq I_0e^\varepsilon \; .
\end{align*}
\end{proof}
Using the super-exponential definition of a sub-Gaussian and sub-exponential random variable \cite{versh}, we say that a random vector $V\in\mathbb{R}^d$ is sub-Gaussian with parameter $\sigma$ if $\mathbb{E}\left[e^{\left(\frac{\langle u,V \rangle}{\sigma}\right)^2}\right] \leq 2$ for any unit vector $u\in\mathbb{R}^d$. We say that a random vector $V\in\mathbb{R}^d$ is sub-exponential with parameter $\sigma$ if $\mathbb{E}\left[e^{\left|\frac{\langle u,V \rangle}{\sigma}\right|}\right] \leq 2$ for any unit vector $u\in\mathbb{R}^d$.
\begin{prop} \label{prop4}
If $S_\theta(X)$ is sub-Gaussian with parameter $\sigma$ and $\varepsilon\geq 1$ then $$\mathsf{Tr}(I_Y(\theta)) \leq 2\sigma^2\varepsilon \; .$$
\end{prop}
\begin{proof}
Using the convexity of $x\mapsto e^{x^2}$,
\begin{align*}
\exp\left(\mathbb{E}\left[\frac{1}{\sigma}\langle u,S_\theta(X)\rangle \frac{Q(y|X)}{\mathbb{E}[Q(y|X)]}\right]^2\right) & \leq \mathbb{E}\left[\exp\left(\left(\frac{\langle u,S_\theta(X)\rangle}{\sigma}\right)^2\right)\frac{Q(y|X)}{\mathbb{E}[Q(y|X)]}\right] \\
& \leq e^\varepsilon\mathbb{E}\left[\exp\left(\frac{\langle u,S_\theta(X)\rangle}{\sigma}\right)^2\right] \\
& \leq 2e^\varepsilon \; .
\end{align*}
Taking logs,
$$\mathbb{E}\left[\langle u,S_\theta(X)\rangle \frac{Q(y|X)}{\mathbb{E}[Q(y|X)]}\right]^2 \leq \sigma^2(\varepsilon + \log 2)$$
so that for $\varepsilon \geq 1$,
$$\mathbb{E}\left[\langle u,S_\theta(X)\rangle \frac{Q(y|X)}{\mathbb{E}[Q(y|X)]}\right]^2 \leq 2\sigma^2\varepsilon \; .$$
\end{proof}
With a nearly identical proof we also have the following sub-exponential result.
\begin{prop} \label{prop5}
If $S_\theta(X)$ is sub-exponential with parameter $\sigma$ and $\varepsilon\geq 1$ then $$\mathsf{Tr}(I_Y(\theta)) \leq 2\sigma^2\varepsilon^2 \; .$$
\end{prop}
\begin{figure}
\begin{center}
\begin{tabular}{ | c | c | c | }
\hline
 condition on $S_\theta(X)$ & upper bound on $\mathsf{Tr}(I_Y(\theta))$ & lower bound for $\ell_2^2$ risk (with $n$ samples) \\
\hline
finite variance $I_0$ & $I_0(e^\varepsilon-1)^2$ & $\frac{d^2}{nI_0(e^\varepsilon-1)^2}$ \\ 
finite variance $I_0$ & $I_0e^\varepsilon$ & $\frac{d^2}{nI_0e^\varepsilon}$ \\ 
$\sigma^2$-sub-Gaussian & $\sigma^2\varepsilon$ & $\frac{d^2}{n\sigma^2\varepsilon}$ \\  
$\sigma$-sub-exponential & $\sigma^2\varepsilon^2$ & $\frac{d^2}{n\sigma^2\varepsilon^2}$ \\
 \hline
\end{tabular}
\end{center}
\caption{A summary of the (order-wise) upper bounds on $\mathsf{Tr}(I_Y(\theta))$ under different conditions on the score function random vector $S_\theta(X)$.}
\label{fig2}
\end{figure}
\subsection{Interactive Models}
So far we have assumed that there are no interactions between the different samples during privatization, so that the privatized samples $Y_1,\ldots,Y_n$ are independent. It is worth pointing out that our Fisher information bounds, and the corresponding lower bounds in the estimation error, can also be extended to interactive communication models where the privatized samples are no longer necessarily independent. We describe both a sequential interactive model \cite{duchi2013local} and a more general fully interactive blackboard model \cite{duchi2019lower} below:
\begin{itemize}
\item[(i)]{\bf Sequential Interaction:} In this scenario, the samples $X_1,\ldots,X_n$ are ordered and the mechanism for sample $i$ can depend on the previously privatized samples $Y_1,\ldots,Y_{i-1}$. Formally, there is a collection of $\varepsilon$ differentially private mechanisms $Q_i(Y_i|x_i,y_1,\ldots,y_{i-1})$ such that
$$\frac{Q_i(S|x_i,y_1,\ldots,y_{i-1})}{Q_i(S|x'_i,y_1,\ldots,y_{i-1})} \leq e^\varepsilon$$
for any $i=1,\ldots,n$, event $S$, and $x_i,x'_i,y_1,\ldots,y_{i-1}$. Using the chain rule for Fisher information,
\begin{equation} \label{eq:chain_rule}
\mathsf{Tr}(I_{Y_1,\ldots,Y_n}(\theta))  = \sum_{i=1}^n \mathbb{E}_{Y_1,\ldots,Y_{i-1}}\left[\mathsf{Tr}(I_{Y_i|Y_{i-1},\ldots,Y_1}(\theta))\right]
\end{equation}
where $I_{Y_i|y_{i-1},\ldots,y_1}(\theta)$ denotes the Fisher information computed using the distribution for $Y_i$ conditioned on $y_1,\ldots,y_{i-1}$. For a given $\theta$, $X_i$ is independent of $Y_1,\ldots,Y_{i-1}$, so conditioning on $y_1,\ldots,y_{i-1}$ just determines which mechanism is used, and because each mechanism satisfies the $\varepsilon$ differentially private condition, each term inside the expectation in \eqref{eq:chain_rule} can be upper bounded just as in Propositions \ref{prop2}-\ref{prop5}, and the total bound on the Fisher information from all $n$ samples will remain the same.
\item[(ii)] {\bf Fully Interactive Blackboard Model:} In this scenario, there are multiple rounds of sequential communication, and a public blackboard with all of the information released after each round is available to all nodes in future rounds of communication. Suppose there are $T$ rounds of communication indexed by $t=1,\ldots,T$. Node $i$ releases $Y_{i,t}$ on round $t$ which is drawn from $Q_{i,t}(Y_{i,t}|b_1,\ldots,b_{t-1},y_{1,t},\ldots,y_{i-1,t},x_i)$ where $B_t = (Y_{1,t},\ldots,Y_{n,t})$ is the blackboard that is visible to all nodes after round $t$. Call the total transcript after all rounds of communication $Z=(B_1,\ldots,B_T)$. We assume that there is a total ``privacy budget'' of $\varepsilon$ for each $X_i$ in the sense that
\begin{equation*}
\frac{\text{Pr}(Z\in S|x_1,\ldots,x_{i},\ldots,x_n)}{\text{Pr}(Z \in S|x_1,\ldots,x'_{i},\ldots,x_n)} \leq e^\varepsilon
\end{equation*}
for any event $S, \; i, \; x_1,\ldots,x_n, \; x_i'$. Note that by the chain rule, the conditional density of $Z$ can be written as
\begin{equation} \label{eq:bb_expansion}
\frac{Q(Z|x_1,\ldots,x_{i},\ldots,x_n)}{Q(Z|x_1,\ldots,x'_{i},\ldots,x_n)} = \frac{\prod_t Q_{i,t}(Y_{i,t}|B_1,\ldots,B_{t-1},Y_{1,t},\ldots,Y_{i-1,t},x_i)}{\prod_tQ_{i,t}(Y_{i,t}|B_1,\ldots,B_{t-1},Y_{1,t},\ldots,Y_{i-1,t},x'_i)} \; .
\end{equation}
The total Fisher information from the transcript $Z$ can be upper bounded by $n$ times the bounds in Propositions \ref{prop2}-\ref{prop5} under the same conditions on the score $S_\theta(X)$, so that the same lower bounds in estimation error also apply for this more general interaction model. The details for this are found in Appendix \ref{app:blackboard}.
\end{itemize}

\section{Applications}
In this section we show how the upper bounds on $\mathsf{Tr}(I_Y(\theta))$ developed in the last section can be used to imply order-optimal lower bounds on the private estimation problem using \eqref{eq:vt2}. We are interested in characterizing the minimax risk
$$\inf_{(Q,\hat\theta)}\max_{\theta\in\Theta} \mathbb{E}\|\hat\theta(Y_1,\ldots,Y_n) - \theta\|_2^2$$
where the estimator $\hat\theta$ is a function of the privatized samples $Y_1,\ldots,Y_n$, and the infimum is taken jointly over the estimator $\hat\theta$ and privatization mechanism $Q$. Upper bounds can therefore be found by jointly designing an estimator and privatization mechanism that achieve a certain worst-case error.

\begin{cor}[Gaussian location model]
Suppose $X_i \sim \mathcal{N}(\theta,\sigma_0^2 I_d)$ and $\Theta = [-B,B]^d$. In this case $S_\theta(X_i)\sim \mathcal{N}\left(0,\frac{1}{\sigma_0^2}I_d\right)$ and the conditions for Proposition 2 and Proposition 4 are satisfied with $I_0 = \frac{1}{\sigma_0^2}$ and $\sigma^2 = O\left(\frac{1}{\sigma_0^2}\right)$, respectively. Using the van Trees inequality we have
$$\max_{\theta\in\Theta} \mathbb{E}\|\hat\theta(Y_1,\ldots,Y_n) - \theta\|_2^2 \geq c\frac{\sigma_0^2 d^2}{n\min\{\varepsilon^2,\varepsilon\}}$$
for an absolute constant $c$ if $\frac{n\min\{\varepsilon^2,\varepsilon\}}{\sigma_0^2} \geq \frac{d}{B^2}$.
\end{cor}
This lower bound for the Gaussian location model matches both the lower and upper bounds detailed in \cite{duchi2019lower}. The condition on $n$ is a mild technical condition that ensures the second term in the denominator of \eqref{eq:vt2} will not dominate the order of the lower bound. We will make similar assumptions in the following examples.
\begin{cor}[discrete distribution estimation] \label{cor2}
Suppose $\mathcal{X} = [1:d+1]$ and $X_i \sim \text{Mult}(1,\theta)$ where $\Theta = \{\theta\in\mathbb{R}^{d+1} \; : \; \sum_{i=1}^{d+1} \theta_i = 1\}.$ Then
$$\max_{\theta\in\Theta} \mathbb{E}\|\hat\theta(Y_1,\ldots,Y_n) - \theta\|_2^2 \geq c\frac{ d}{n\min\{(e^\varepsilon-1)^2,e^\varepsilon\}}$$
for an absolute constant $c$ if $n\min\{(e^\varepsilon-1)^2,e^\varepsilon\} \geq d^2$.
\end{cor}
Corollary \ref{cor2} follows by applying Propositions \ref{prop2} and \ref{prop3} with variance $I_0 \leq 6d$. The details are included in Appendix \ref{app3}.
The lower bound for the discrete distribution example gives the same order as the upper and lower bounds from \cite{duchi2013local} when $\varepsilon$ is close to zero, and it also matches the upper and lower bounds from \cite{yebarg2018} when $1 \preceq e^\varepsilon \preceq d$. Other mechanisms for this model are discussed in \cite{kairouz16}.
%\begin{cor}[Sparse Bernoulli models] \label{cor3}
%Suppose $P_\theta = \prod_{j=1}^d\mathsf{Bern}(\theta_j)$.
%\begin{itemize}
%%\item[(i)]Dense Bernoulli: If $\Theta = [0,1]^d$ and $n\min\{\varepsilon^2,\varepsilon\} \geq d$, then
%%$$\inf_{(Q,\hat\theta)}\max_{\theta\in\Theta} \mathbb{E}\|\hat\theta(Y_1,\ldots,Y_n) - \theta\|_2^2 \asymp \frac{d^2}{n\min\{\varepsilon^2,\varepsilon,d\}} \; .$$
%\item[(i)]Sparse Bernoulli: If $\Theta = \left\{\theta\in[0,1]^d\; : \; \sum_{j=1}^d\theta_j = 1\right\}$ and $n\min\{(e^\varepsilon-1)^2,e^\varepsilon\} \geq d^2$, then
%$$\inf_{(Q,\hat\theta)}\max_{\theta\in\Theta} \mathbb{E}\|\hat\theta(Y_1,\ldots,Y_n) - \theta\|_2^2 \asymp \frac{ d}{n\min\{(e^\varepsilon-1)^2,e^\varepsilon,d\}} \; .$$
%\item[(ii)]$s$-Sparse Bernoulli:  If $\Theta = \left\{\theta\in[0,1]^d\; : \; \sum_{j=1}^d\theta_j \leq s \right\}$, $s\leq d^{1-\delta}$ for some $\delta>0$, $n\varepsilon \geq d\log d$, and $20\log d \leq \varepsilon \leq s\log d$, then
%\begin{equation} \label{eq:cor3}\inf_{(Q,\hat\theta)}\max_{\theta\in\Theta} \mathbb{E}\|\hat\theta(Y_1,\ldots,Y_n) - \theta\|_2^2 \asymp \frac{s^2\log d}{n\varepsilon} \; .\end{equation}
%\end{itemize}
%\end{cor}
\begin{cor}[Sparse Bernoulli models] \label{cor3}
Suppose $P_\theta = \prod_{j=1}^d\mathsf{Bern}(\theta_j)$.
\begin{itemize}

\item[(i)]Sparse Bernoulli: If $\Theta = \left\{\theta\in[0,1]^d\; : \; \sum_{j=1}^d\theta_j \leq 1\right\}$ and $n\min\{(e^\varepsilon-1)^2,e^\varepsilon\} \geq d^2$, then
$$\inf_{(Q,\hat\theta)}\max_{\theta\in\Theta} \mathbb{E}\|\hat\theta(Y_1,\ldots,Y_n) - \theta\|_2^2 \asymp \frac{ d}{n\min\{(e^\varepsilon-1)^2,e^\varepsilon,d\}} \; .$$

\item[(ii)]$s$-Sparse Bernoulli:  If $\Theta = \left\{\theta\in[0,1]^d\; : \; \sum_{j=1}^d\theta_j \leq s \right\}$, $s\leq d^{1-\delta}$ for some $\delta>0$. 
\begin{itemize}

    \item[$\bullet$] High privacy regime: if $ \varepsilon \leq \log \lp d/s \rp$ and 
    $$ \frac{n}{\log n} \geq \frac{20}{\min\lp \varepsilon^2, 1 \rp}d^3\log d $$
    then
    \begin{equation} \label{eq:cor3}\inf_{(Q,\hat\theta)}\sup_{\theta\in\Theta} \mathbb{E}\|\hat\theta(Y_1,\ldots,Y_n) - \theta\|_2^2 \asymp \frac{s d}{n\min\lp e^\varepsilon, \lp e^\varepsilon -1 \rp^2, d \rp} \; .\end{equation}
    
    \item[$\bullet$] Low privacy regime: if $20\log d \leq \varepsilon \leq s\log d$ and $n\varepsilon \geq d\log d$, then
    \begin{equation} \label{eq:cor3}\inf_{(Q,\hat\theta)}\sup_{\theta\in\Theta} \mathbb{E}\|\hat\theta(Y_1,\ldots,Y_n) - \theta\|_2^2 \asymp \frac{s^2\log d}{n\varepsilon} \; .\end{equation}

\end{itemize}
\end{itemize}
\end{cor}
We sketch the proof of Corollary \ref{cor3} in the next section. The sparse Bernoulli example is illustrative of several interesting phenomena. In the sparse Bernoulli model (i), the minimax risk scales exactly the same as in the discrete distribution estimation problem. This shows that even in a statistical model with independent components, the dependence on $\varepsilon$ can be exponential instead of linear.  In this way, the scaling is dictated by the properties of the score function $S_\theta(X)$ rather than the independence of the model. In the $s$-sparse Bernoulli model (ii), we see that in the privacy range  $\log d \preceq \varepsilon \preceq s\log d$, the minimax risk scales as $\frac{s^2\log d}{n\varepsilon}$ instead of the centralized (i.e. without privacy constraints) rate $\frac{s}{n}$. This means that the sample size penalty is of the order $\frac{s\log d}{\varepsilon}$ for privacy. This is noteworthy in that the penalty scales linearly only with the sparsity $s$, rather than with the underlying dimension $d$, which is the case for, e.g, the sparse Gaussian location model.

Note that cases (i) and (ii) above focus on two different parameter regimes of the same model, and the domain $\Theta$ in (i) is a subset of that in (ii). As such, the lower bound in (i) can be regarded as a more local minimax risk bound, while the one in (ii) with $s$ close to $d$ can be regarded as a worst-case minimax bound. These two bounds together illustrate how having additional information that restricts the range of the parameter as in (i) can change the dependence of the risk on the privacy parameter $\varepsilon$.

\section{Proof of Corollary \ref{cor3}}
\subsection{Sparse Bernoulli (i)}
The lower bound follows by applying Propositions \ref{prop2} and \ref{prop3} with a score function that satisfies $I_0 \leq 3d$, as we check below. We restrict our attention to $\Theta' = \left[\frac{1}{2d},\frac{1}{d}\right]^d\subset\Theta$ using the fact that $\sup_{\theta\in\Theta} \mathbb{E}\|\hat\theta(Y_1,\ldots,Y_n) - \theta\|_2^2 \geq \sup_{\theta\in\Theta'} \mathbb{E}\|\hat\theta(Y_1,\ldots,Y_n) - \theta\|_2^2 \; .$ The score function for each component is
\begin{equation} \label{eq:bern_score} S_{\theta_j}(x_j) = \frac{\partial}{\partial\theta_j}\log f(x_j|\theta_j) = \begin{cases} \frac{1}{\theta_i} & \; \text{, if } \; x_j=1 \\ \frac{-1}{1-\theta_i} & \; \text{, if } \; x_j =0 \; .\end{cases} \end{equation}
so that the variance of each component is
\begin{align*}
\mathbb{E}[S_{\theta_j}(x_j)^2] & = \theta_j\frac{1}{\theta_j^2} + (1-\theta_j)\frac{1}{(1-\theta_j)^2} = \frac{1}{\theta_j} + \frac{1}{(1-\theta_j)} \\
& \leq 3d \; .
\end{align*}
By taking sums of independent variables we also have $\mathbb{E}[\langle u, S_{\theta}(x) \rangle^2] \leq 3d$ for any unit vector $u\in\mathbb{R}^d$ and $\theta\in\Theta'$, as desired.

\input{sec_low_sparsity.tex}

\subsection{$s$-Sparse Bernoulli (ii)}
\subsubsection{High privacy regime}
The lower bound follows in the same way as that of $(i)$ above, except focusing on $\frac{s}{2d}\leq\theta_j\leq\frac{s}{d}$ for each $j=1,\ldots,d$. For the upper bound, we describe a scheme that achieves this error but requires at least sequential interaction between the nodes.
\input{sec_s_sparse_high_privacy.tex}
\subsubsection{Low privacy regime}
\input{sec_general_sparsity.tex}

\input{sec_general_sparsity_app.tex}

\section*{Acknowledgements}
This work was supported in part by NSF award CCF-1704624 and by the Center for Science of Information (CSoI), an NSF Science and Technology Center, under grant agreement CCF-0939370.

\bibliographystyle{IEEEtran}
\bibliography{../di.bib}

\newpage

\begin{appendix}

\subsection{Proof of Proposition \ref{prop1}} \label{app1}
Fix some $\delta>0.$ Suppose, for contradiction, that $$\frac{Q(y|x)}{Q(y|x')} > e^\varepsilon+\delta$$ for all $y\in S$ with $\nu(S)>0$. We have
\begin{align*}
\frac{Q(S|x)}{Q(S|x')} = \frac{\int_S Q(y|x)d\nu(y)}{\int_S Q(y|x')d\nu(y)} \geq \inf_{y\in S} \frac{Q(y|x)}{Q(y|x')} \geq e^\varepsilon + \delta \; .
\end{align*}
This contradicts $Q(\cdot|\cdot)$ being an $\varepsilon$-differentially private mechanism, and thus we must have $$\frac{Q(y|x)}{Q(y|x')}\leq e^\varepsilon + \delta$$
for $\nu$-almost all $y$. Taking $\delta \to 0$ and using the measure's continuity from above completes the proof.

\subsection{Proof of Lemma \ref{lem1}} \label{app2}
\begin{align} \label{eq:lem1}
\mathsf{Tr}(I_Y(\theta)) & = \sum_{i=1}^d \mathbb{E}\left[\left(\frac{\partial}{\partial\theta_i}\log f(Y|\theta)\right)^2\right] \nonumber\\
& = \sum_{i=1}^d \mathbb{E}\left[\left(\frac{\frac{\partial}{\partial\theta_i}f(Y|\theta)}{f(Y|\theta)}\right)^2\right] \nonumber\\
& = \sum_{i=1}^d \mathbb{E}\left[\left(\frac{\int Q(Y|x)\frac{\partial}{\partial\theta_i} f(x|\theta)d\mu(x)}{f(Y|\theta)}\right)^2\right] \\
& = \sum_{i=1}^d \mathbb{E}\left[\left(\int\frac{Q(Y|x)f(x|\theta)}{f(Y|\theta)}\frac{\frac{\partial}{\partial\theta_i} f(x|\theta)}{f(x|\theta)}d\mu(x)\right)^2\right] \nonumber \\
& = \sum_{i=1}^d \mathbb{E}_Y \mathbb{E}_X\left[\frac{\partial}{\partial\theta_i}\log f(x|\theta) \bigg| Y\right]^2  \nonumber \\
& = \mathbb{E}_Y\|\mathbb{E}_X[S_\theta(X)|Y]\|_2^2 \nonumber \; .
\end{align}
The key step \eqref{eq:lem1} relies on interchanging integration over the sample space and differentiation with respect to the components $\theta_j$ which can be made precise via Lebesgue's Dominated Convergence Theorem as shown in Appendix \ref{app:reg} regarding regularity conditions.

\subsection{Proof of Corollary \ref{cor2}} \label{app3}
Without loss of generality we focus on a subset $\Theta'\subset\Theta$ defined by $$\Theta' = \left\{\theta\in\Theta \; : \; \frac{1}{4d}\leq\theta_i\leq\frac{1}{2d} \text{ for } i=1,\ldots,d\right\} \; ,$$
and only consider the error from the first $d$ components of $\theta_i$. 
We can do this because
$$\max_{\theta\in\Theta} \mathbb{E}\|\hat\theta(Y_1,\ldots,Y_n) - \theta\|_2^2 \geq \max_{\theta\in\Theta'}\left[ \mathbb{E}\sum_{i=1}^d\left(\hat\theta_i-\theta_i\right)^2\right] \; .$$
It remains to show that for all $\theta\in\Theta'$ and unit vectors $u\in\mathbb{R}^d$, $$\mathbb{E}[\langle u,S_\theta(X)\rangle^2] \leq 6d$$
where
\begin{align*}
S_\theta(X) & = (S_{\theta_1}(x),\ldots,S_{\theta_d}(x)) \\
& = \left(\frac{\partial}{\partial\theta_1}\log f(x|\theta),\ldots,\frac{\partial}{\partial\theta_d}\log f(x|\theta)\right)
\end{align*}
is the score from just the first $d$ components.
To see this note that
$$\theta_{d+1} = 1-\sum_{i=1}^{d} \theta_i, $$
and
$$S_{\theta_i}(x) = \begin{cases} \frac{1}{\theta_i}, & \; x=i \\ -\frac{1}{\theta_{d+1}}, & \; x=d+1 \\ 0, & \; \text{otherwise} \end{cases}$$
for $i=1,\ldots,d$. Then for any unit vector $u=(u_1,\ldots,u_{d})$,
\begin{align*}
\mathbb{E}[\langle u,S_\theta(X)\rangle^2] & = \sum_{x=1}^{d+1} \theta_x \left( \sum_{i=1}^{d}u_iS_{\theta_i}(x)\right)^2 \nonumber \\
& = \theta_{d+1}\frac{1}{\theta_{d+1}^2}\left(\sum_{i=1}^d u_i \right)^2 + \sum_{x=1}^{d} \theta_x \left( \sum_{i=1}^{d}u_iS_{\theta_i}(x)\right)^2 \nonumber \\
& \leq 2d + \sum_{x=1}^d \theta_x u_x^2\frac{1}{\theta_x^2} \leq 6d \; .
\end{align*}
The corollary then follows by applying Propositions \ref{prop2} and \ref{prop3} with $I_0 = 6d$ to equation \eqref{eq:vt2}.

%\subsection{Sparse Bernoulli upper bound}
%\label{lower_bound_app}
%\input{../sec_low_sparsity_app.tex}
%\subsection{$s$-Sparse Bernoulli upper bound}
%\label{s_lower_bound_app}
%\input{../sec_general_sparsity_app.tex}

\subsection{Regularity Conditions}\label{app:reg}
We make the following assumptions on the statistical model $P_\theta$:
\begin{itemize}
    \item[(i)] The density $f(x|\theta)$ is such that $\sqrt{f(x|\theta)}$ is continuously differentiable with respect to $\theta_j$ for $j=1,\ldots,d$ and $\mu$-almost all $x\in\mathcal{X}$. Note that this is the same as assuming that the density $f(x|\theta)$ itself is continuously differentiable if we assume that $f(x|\theta)>0$, and this positivity assumption can always be made valid by considering all integrals to only be over the subset of $\mathcal{X}$ with $f(x|\theta)>0$.
    \item[(ii)] The Fisher information for each component $I_X(\theta_j) = \mathbb{E}\left[\left(\frac{\partial}{\partial\theta_j}\log f(x|\theta)\right)^2\right]$ exists and is a continuous function of $\theta_j$ for each $j=1,\ldots,d$.
\end{itemize}
It can easily be checked that for the Gaussian location model, discrete distribution estimation, and sparse Bernoulli models these conditions are met for an appropriate subset of the space of possible parameter values $\Theta$.

These conditions are relatively standard sufficient conditions for a statistical model to be \emph{differentiable in quadratic mean} \cite{Vandervaart2000}. Unfortunately the differentiable in quadratic mean condition itself is not appropriate for developing Cram\'er-Rao type lower bounds, and so it will not work for our purposes. One important aspect of these conditions is that we make assumptions on the statistical model $P_\theta$, but not $Q_\theta$, so that there are no implicit assumptions on the privacy mechanism.

\begin{lem} \label{lem:reg}
Under the conditions above, $f(y|\theta)$ is continuously differentiable with respect to $\theta_j$ and
\begin{align*}
\frac{\partial}{\partial\theta_j}f(y|\theta) & = \frac{\partial}{\partial\theta_j}\int Q(y|x)f(x|\theta)d\mu(x) \\
& = \int Q(y|x)\frac{\partial}{\partial\theta_j}f(x|\theta)d\mu(x)
\end{align*}
at $\nu$-almost any $y$.
\end{lem}
\begin{proof}
For simplicity we consider the scalar case with $d=1$. The proof for each component of the vector case is identical. Without yet knowing if the limit exists, formally we have
\begin{align*}
\frac{\partial}{\partial\theta}f(y|\theta) & = \lim_{h\to0}\frac{1}{h}\left(\int Q(y|x)f(x|\theta+h)d\mu(x) - \int Q(y|x)f(x|\theta)d\mu(x)\right) \\
& = \lim_{h\to0} \int Q(y|x)\int_0^1 f'(x|\theta+hu)du \, d\mu(x) \\
& = \lim_{h\to0} 2\int \int_0^1 Q(y|x)\sqrt{f(x|\theta+hu)}\sqrt{f(x|\theta+hu)}'du \, d\mu(x) \; .
\end{align*}
For each $h$ define the set
$$A_h = \left\{ x\in\mathcal{X} \; : \; \sup_{v:|\theta-v|<h} \sqrt{f(x|v)} < 2\sqrt{f(x|\theta)} \quad , \quad \sup_{v:|\theta-v|<h} |\sqrt{f(x|v)}'| < 2|\sqrt{f(x|\theta)}'| \right\}$$
and split the integral into two terms, considering the set $A_h$ and its complement $A_h^C$ separately:
\begin{align}
\int\int_0^1 & Q(y|x)\sqrt{f(x|\theta+hu)}\sqrt{f(x|\theta+hu)}'du \, d\mu(x) \\
= & \int_{A_h}\int_0^1 Q(y|x)\sqrt{f(x|\theta+hu)}\sqrt{f(x|\theta+hu)}'dud\mu(x) \label{reg_term1}\\
& + \int_{A_h^C}\int_0^1 Q(y|x)\sqrt{f(x|\theta+hu)}\sqrt{f(x|\theta+hu)}'dud\mu(x). \label{reg_term2}
\end{align}
To deal with term \eqref{reg_term1} we can use Lebesgue's dominated convergence theorem noting that
\begin{equation}\label{reg_term3}|1_{A_h}(x)Q(y|x)\sqrt{f(x|\theta+hu)}\sqrt{f(x|\theta+hu)}'| \leq 4|Q(y|x)\sqrt{f(x|\theta)}\sqrt{f(x|\theta)}'| \; .
\end{equation}
The right-hand side of display \eqref{reg_term3} is absolutely integrable by the Cauchy-Schwarz inequality:
\begin{align*}
\int |Q(y|x)\sqrt{f(x|\theta)}\sqrt{f(x|\theta)}'|d\mu(x) & \leq \left(\int Q(y|x)^2f(x|\theta)d\mu(x)\right)^\frac{1}{2}\left(\int \frac{f'(x|\theta)^2}{f(x|\theta)}d\mu(x)\right)^\frac{1}{2} \\
& \leq f(y|\theta)\sqrt{e^\varepsilon I_X(\theta)} \; .
\end{align*}
This allows us to switch the limit inside the integral to get
$$\lim_{h\to0} \int_{A_h^C}\int_0^1 Q(y|x)\sqrt{f(x|\theta+hu)}\sqrt{f(x|\theta+hu)}'dud\mu(x) = \int Q(y|x)\sqrt{f(x|\theta)}\sqrt{f(x|\theta)}'d\mu(x)$$
where we have used the continuity of $\sqrt{f(x|\theta)}$ and $\sqrt{f(x|\theta)}'$ to see that $1_{A_h}(x)\to 1$.

It remains to show that term \eqref{reg_term2} approaches zero as $h\to 0$. For this we again use the Cauchy-Schwarz inequality:
\begin{align}
\int_{A_h^C}\int_0^1 & |Q(y|x)\sqrt{f(x|\theta+uh)}\sqrt{f(x|\theta+uh)}'|dud\mu(x) \nonumber\\
& \leq \left(\int\int_0^1 Q(y|x)^2f(x|\theta+uh)dud\mu(x)\right)^\frac{1}{2}\left(\int_{A_h^C}\int_0^1 \frac{f'(x|\theta+uh)^2}{f(x|\theta+uh)}dud\mu(x)\right)^\frac{1}{2} \nonumber\\ \label{reg_term4}
& \leq e^\varepsilon f(y|\theta) \left(\int_{A_h^C}\int_0^1 \frac{f'(x|\theta+uh)^2}{f(x|\theta+uh)}dud\mu(x)\right)^\frac{1}{2} \; .
\end{align}
The term inside the parentheses in \eqref{reg_term4} goes to zero as $h\to 0$ since
\begin{align*}
I_X(\theta) & = \lim_{h\to0}\left( \int_{A_h}\int_0^1 \frac{f'(x|\theta+uh)^2}{f(x|\theta+uh)}dud\mu(x) + \int_{A_h^C}\int_0^1 \frac{f'(x|\theta+uh)^2}{f(x|\theta+uh)}dud\mu(x)\right)
\end{align*}
and
$$\int_{A_h}\int_0^1 \frac{f'(x|\theta+uh)^2}{f(x|\theta+uh)}dud\mu(x) \to I_X(\theta)$$
as $h\to 0$ using the dominated convergence theorem just as above.
\end{proof}
\subsubsection{Applying the van Trees Inequality}
In order to apply the van Trees inequality we will need
$$\int \left(\frac{\partial}{\partial\theta_j}\log f(y|\theta)\right)f(y|\theta)d\nu(y) = \int \frac{\partial}{\partial\theta_j} f(y|\theta)d\nu(y) = 0$$
for each $j=1,\ldots,d$. In this subsection we check this condition under assumptions (i) and (ii) above regarding the distributions $P_\theta$ and their densities $f(x|\theta)$. We will make no further assumptions on $f(y|\theta)$ so that there are no implicit assumptions on the privacy mechanism $Q$ other than it being a regular conditional distribution and an $\varepsilon$-differentially private mechanism.

Using Lemma \ref{lem:reg} and the Fubini-Tonelli Theorem,
\begin{align}
\int \frac{\partial}{\partial\theta_j} f(y|\theta)d\nu(y) = & \int \frac{\partial}{\partial\theta_j} \int Q(y|x)f(x|\theta)d\mu(x)d\nu(y) \nonumber\\
 = & \int \int Q(y|x)\frac{\partial}{\partial\theta_j} f(x|\theta)d\mu(x)d\nu(y) \nonumber\\
 = & \int \int_{\left\{x:\frac{\partial}{\partial\theta_j} f(x|\theta)\geq 0\right\}} Q(y|x)\frac{\partial}{\partial\theta_j} f(x|\theta)d\mu(x)d\nu(y) \nonumber\\
& + \int \int_{\left\{x:\frac{\partial}{\partial\theta_j} f(x|\theta) < 0\right\}} Q(y|x)\frac{\partial}{\partial\theta_j} f(x|\theta)d\mu(x)d\nu(y) \nonumber \\
 = & \int_{\left\{x:\frac{\partial}{\partial\theta_j} f(x|\theta)\geq 0\right\}} \frac{\partial}{\partial\theta_j} f(x|\theta)d\mu(x) + \int_{\left\{x:\frac{\partial}{\partial\theta_j} f(x|\theta) < 0\right\}} \frac{\partial}{\partial\theta_j} f(x|\theta)d\mu(x) \nonumber\\
 = & \frac{\partial}{\partial\theta_j} \int f(x|\theta)d\mu(x)  = 0 \; . \nonumber
\end{align}

\subsection{Blackboard Model} \label{app:blackboard}
The density of the total transcript $Z$ can be written as
\begin{align*}
f(z|\theta) & = \mathbb{E}_{X_1,\ldots,X_n}\left[\prod_{i,t} Q_{i,t}(y_{i,t}| b_1,\ldots,b_{t-1},y_{1,t},\ldots,y_{i-1,t},X_i)\right] \\
& = \prod_{i=1}^n\mathbb{E}_{X_1,\ldots,X_n}\left[\prod_{t} Q_{i,t}(y_{i,t}| b_1,\ldots,b_{t-1},y_{1,t},\ldots,y_{i-1,t},X_i)\right] \\
& = \prod_{i=1}^n\mathbb{E}_{X_i}\left[p_{i,z}(X_i)\right]
\end{align*}
where $$p_{i,z}(x_i) = \prod_{t} Q_{i,t}(y_{i,t}| b_1,\ldots,b_{t-1},y_{1,t},\ldots,y_{i-1,t},x_i) \; .$$
The score for this total transcript has components
$$\frac{\partial}{\partial\theta_j} \log f(z|\theta) = \sum_{i=1}^n \frac{\mathbb{E}_{X_i}\left[S_{\theta_j}(X_i)p_{i,z}(X_i)\right]}{\mathbb{E}_{X_i}\left[p_{i,z}(X_i)\right]} \; .$$
To get the above display we require interchanging differentiation and integration just like in the proof of Lemma \ref{lem1}. The trace of the Fisher information from the whole transcript is thus
\begin{align}
\mathsf{Tr}(I_Z(\theta)) & = \sum_{j=1}^d \mathbb{E}_Z\left[\left(\frac{\partial}{\partial\theta_j}\log f(Z|\theta)\right)^2\right] \nonumber \\
& = \mathbb{E}_Z\left[ \sum_{i,j}\left(\frac{\mathbb{E}_{X_i}\left[S_{\theta_j}(X_i)p_{i,Z}(X_i)\right]}{\mathbb{E}_{X_i}\left[p_{i,Z}(X_i)\right]}\right)^2 \label{eq:bb_term1}\right]
\end{align}
where \eqref{eq:bb_term1} follows because the cross terms
\begin{align*}
\mathbb{E}_Z & \left[ \frac{\mathbb{E}_{X_i} \left[S_{\theta_j}(X_i)p_{i,Z}(X_i)\right]}{\mathbb{E}_{X_i}\left[p_{i,Z}(X_i)\right]}  \frac{\mathbb{E}_{X_k}\left[S_{\theta_k}(X_k)p_{k,Z}(X_k)\right]}{\mathbb{E}_{X_k}\left[p_{k,Z}(X_k)\right]}\right]\\
 & = \mathbb{E}_{X_i,X_k}\left[S_{\theta_j}(X_i)S_{\theta_j}(X_j)\right] = 0
\end{align*}
for $i\neq k$.
\subsubsection{Blackboard Proposition 2}
Let
$$u_Z = \frac{\mathbb{E}_{X_i}\left[S_{\theta}(X_i)\frac{p_{i,Z}(X_i)}{\mathbb{E}_{X_i}\left[p_{i,Z}(X_i)\right]}\right]}{\left\|\mathbb{E}_{X_i}\left[S_{\theta}(X_i)\frac{p_{i,Z}(X_i)}{\mathbb{E}_{X_i}\left[p_{i,Z}(X_i)\right]}\right]\right\|_2} \; .$$
Following from \eqref{eq:bb_term1},
\begin{align*}
\mathsf{Tr}(I_Z(\theta)) & =  \sum_{i=1}^n\mathbb{E}_Z\left[\bigg\langle u_Z, \mathbb{E}_{X_i}\left[S_{\theta}(X_i)\frac{p_{i,Z}(X_i)}{\mathbb{E}_{X_i}\left[p_{i,Z}(X_i)\right]}\right]\bigg\rangle^2\right] \\
& = \sum_{i=1}^n\mathbb{E}_Z\left[\mathbb{E}_{X_i}\left[\big\langle u_Z, S_{\theta}(X_i)\big\rangle\frac{p_{i,Z}(X_i)}{\mathbb{E}_{X_i}\left[p_{i,Z}(X_i)\right]}\right]^2\right] \; .
\end{align*}
We split up the expectation over $X_i$ as follows:
\begin{align*}
\mathbb{E}_{X_i}&\left[\big\langle u_Z, S_{\theta}(X_i)\big\rangle\frac{p_{i,Z}(X_i)}{\mathbb{E}_{X_i}\left[p_{i,Z}(X_i)\right]}\right] \\
\leq & \frac{1}{\min_x p_{i,Z}(x)}\mathbb{E}_{X_i}\left[\big\langle u_Z, S_{\theta}(X_i)\big\rangle p_{i,Z}(X_i)\right] \\
= &  \frac{1}{\min_x p_{i,Z}(x)}\int_{\{x: \langle u_Z,S_\theta(x)\rangle\geq0\}}\big\langle u_Z, S_{\theta}(x)\big\rangle p_{i,Z}(x)f(x|\theta)d\mu(x) \\
& +  \frac{1}{\min_x p_{i,Z}(x)}\int_{\{x: \langle u_Z,S_\theta(x)\rangle<0\}}\big\langle u_Z, S_{\theta}(x)\big\rangle p_{i,Z}(x)f(x|\theta)d\mu(x) \\
 \leq & (e^\varepsilon -1) \int_{\{x: \langle u_Z,S_\theta(x)\rangle\geq0\}}\big\langle u_Z, S_{\theta}(x)\big\rangle f(x|\theta)d\mu(x)
\end{align*}
so that
$$\mathbb{E}_{X_i}\left[\big\langle u_Z, S_{\theta}(X_i)\big\rangle\frac{p_{i,Z}(X_i)}{\mathbb{E}_{X_i}\left[p_{i,Z}(X_i)\right]}\right]^2 \leq I_0(e^\varepsilon-1)^2$$
and
$$\mathsf{Tr}(I_Z(\theta)) \leq nI_0(e^\varepsilon-1)^2$$
as desired.
\subsubsection{Blackboard Proposition 3}
Using Jensen's inequality,
\begin{align*}
\sum_{i=1}^n\mathbb{E}_Z & \left[\mathbb{E}_{X_i}\left[\big\langle u_Z, S_{\theta}(X_i)\big\rangle\frac{p_{i,Z}(X_i)}{\mathbb{E}_{X_i}\left[p_{i,Z}(X_i)\right]}\right]^2\right] \\
& \leq \sum_{i=1}^n\mathbb{E}_Z\left[\mathbb{E}_{X_i}\left[\big\langle u_Z, S_{\theta}(X_i)\big\rangle^2\frac{p_{i,Z}(X_i)}{\mathbb{E}_{X_i}\left[p_{i,Z}(X_i)\right]}\right]\right] \\
& \leq nI_0e^\varepsilon
\end{align*}
where the last step uses \eqref{eq:bb_expansion} and the blackboard differential privacy condition.
\subsubsection{Blackboard Proposition 4} By the convexity of $x\mapsto e^{x^2}$,
\begin{align*}
\exp\left(\left(\frac{\mathbb{E}_{X_i}\left[\langle u_Z, S_{\theta}(X_i)\rangle\frac{p_{i,Z}(X_i)}{\mathbb{E}_{X_i}\left[p_{i,Z}(X_i)\right]}\right]}{\sigma}\right)^2\right) & \leq \mathbb{E}_{X_i}\left[\frac{p_{i,Z}(X_i)}{\mathbb{E}_{X_i}\left[p_{i,Z}(X_i)\right]} \exp\left(\left(\frac{\langle u_Z,S_{\theta}(X_i)\rangle}{\sigma}\right)^2\right) \right] \\
& \leq 2e^\varepsilon \; .
\end{align*}
Taking logs,
$$\left(\mathbb{E}_{X_i}\left[\langle u_Z, S_{\theta}(X_i)\rangle\frac{p_{i,Z}(X_i)}{\mathbb{E}_{X_i}\left[p_{i,Z}(X_i)\right]}\right]\right)^2 \leq \sigma^2(\varepsilon + \log 2) \; .$$
Blackboard Proposition 5 follows in the same way.
\end{appendix}

\end{document}

%% file: macro_wnchen.tex
\newcommand{\lp}{\left(}
\newcommand{\rp}{\right)}
\newcommand{\lb}{\left[}
\newcommand{\rb}{\right]}
\newcommand{\lbp}{\left\{}
\newcommand{\rbp}{\right\}}
\newcommand{\lba}{\left\lvert}
\newcommand{\rba}{\right\rvert}

\newcommand{\mcal}{\mathcal}

\newcommand{\ra}{\rightarrow}

\newcommand{\eqDef}{\triangleq}

\newcommand{\E}{\mathbb{E}}
\newcommand{\Var}{\mathsf{Var}}

\newcommand{\Ber}{\mathrm{Ber}}

%% file: sec_low_sparsity.tex
For the upper bound, we demonstrate an estimator that works by reducing the problem to the discrete distribution estimation problem as follows. We perform the analysis for $\sum_i \theta_i = 1$, but the mechanism and the derivation also hold for $\sum_i \theta_i \leq 1$. Moreover, for any \emph{constant} sparsity, say $\sum_i \theta_i \leq c$, we can always perform a randomized mapping to $\lceil c \rceil \cdot d$ symbols and reduce the problem to $1$-Sparse Bernoulli problems (with $\lceil c \rceil$ repetition). Therefore the result holds for any constant sparsity case, i.e. $\sum_i \theta_i \leq c$.
To convert the product Bernoulli model into the distribution estimation problem, define the mapping 
$$ f(   {X}_k) = \begin{cases}
i, &\text{ if } \lVert    {X}_k \rVert_1 = 1 \text{ and }    {X}_k(i) = 1, \\
d+1, &\text{ if  } \lVert    {X}_k \rVert_1 \neq 1.
\end{cases} $$
Then $f(   {X}_k)$ follows $(p_1,..., p_{d+1})$, with
$$ p_i = \theta_i \cdot\prod_{j\neq i} (1-\theta_j), \,\, \forall i \in [d], \,\,  p_{d+1} = 1-\sum_{j=1}^d p_j.$$ 
Also define $P_S \eqDef \prod_{j=1}^d (1-\theta_j)$, then we have 
$$ \theta_i = \frac{\theta_i \cdot \prod_{j\neq i} (1-\theta_j)}{\theta_i \cdot \prod_{j\neq i} (1-\theta_j)+ \prod_{j=1}^d (1-\theta_j)}= \frac{p_i}{p_i+P_S}. $$
Therefore, our strategy is to estimate $p_i$ and $P_S$ separately, and the final estimator will be 
$$ \hat{\theta}_i\eqDef \frac{ \hat{p}_i }{\hat{p}_i+ \hat{P}_S}. $$
It remains to complete the descriptions of the estimators $\hat p_i$ and $\hat P_S$ and to analyze the error from this strategy.

For ease of analysis, we assume that $\theta_i \leq \frac{1}{2}$. for all $i \in \{1,...,d\}$. Note that this assumption can be easily circumvented by using a randomized mapping  $h:    {X}_k \ra \{0,1\}^{2d}$, such that if $   {X}_k(i) = 1$ then set $h(   {X}_k)_{2i}$ and $h(   {X}_k)_{2i+1}$ to $1$ with probability $\frac{1}{2}$, and otherwise set them to $0$.
Obviously $h(   {X}_k)$ follows product Bernoulli with parameter $\frac{\theta_i}{2} \leq \frac{1}{2}$ in each dimension.

\subsection*{Discrete distribution estimator under LDP}
First we review the distribution estimation problem with LDP constraint, where each node observes a sample $X_k \in \mcal{X} = \lbp 1,...,d\rbp$ from a discrete distribution $   {p} = (p_1,...,p_d)$ and is allowed to transmit information under $\varepsilon$-local privacy constraint. %\cite{ye2017optimal} proposes the following rate-optimal privatization scheme that maps each $X_k$ into $y\in\mcal{Y}_{d,s}\eqDef \lbp y\in \lbp0,1\rbp^d: \sum_i y_i = s\rbp$:
%$$ \mcal{M}_{\text{DE}}(y|i) = \frac{e^\varepsilon y_i + (1-y_i)}{e^\varepsilon {d-1 \choose s-1}+{d-1 \choose s}}. $$
%By selecting $s = \lb\frac{1}{e^{\varepsilon}+1}\rb$, the $\ell_2$ estimation error is
The LDP mechanisms can be viewed as a pair of 
\begin{itemize}
    \item locally privatization mapping $Q_{\text{DE}}(y|i)$ that maps each observation $X_k$ to $Y_k \in \mcal{Y}$
    \item an estimator $\hat{   {p}}\lp Y^n\rp = (\hat{p}_1\lp Y^n \rp,...,\hat{p}_k \lp Y^n\rp)$.
\end{itemize}
In particular, \cite{yebarg2018} propose the following privatization mapping that maps each $X_k$ into $y\in\mcal{Y}_{d, w}\eqDef \lbp y\in \lbp0,1\rbp^d: \sum_i y_i = w\rbp$ with the following transitional probability:
$$ Q_{\text{DE}}(y|i) = \frac{e^\varepsilon y_i + (1-y_i)}{e^\varepsilon {d-1 \choose w-1}+{d-1 \choose w}}. $$
%By selecting $s = \lb\frac{1}{e^{\varepsilon}+1}\rb$, the $\ell_2$ estimation error is
The estimator is 
$$ \lp \frac{(d-1)e^\varepsilon +\frac{(d-1)(d-w)}{w}}{(d-w)(e^\varepsilon-1)}\rp \frac{T_i}{n}-\frac{(w-1)e^\varepsilon+d-w}{(d-w)e^\varepsilon -1},$$
where $T_i \eqDef \sum_{k=1}^n Y_k(i)$.

\begin{theorem}[Proposition~III.1 \citep{yebarg2018}] 
	\begin{equation}\label{eq:l_2_error}
		\E\lVert \hat{   {p}}\lp Y^n\rp -    {p} \rVert^2_2 =
		\frac{1}{n}\lp \frac{\lp w(d-2)+1 \rp e^{2\varepsilon}}{(d-w)\lp e^\varepsilon -1\rp^2}+\frac{2(d-2)}{\lp e^\varepsilon -1\rp^2}+ \frac{(d-2)(d-w)+1}{w\lp e^\varepsilon -1\rp^2} - \sum_i p_i^2\rp.
	\end{equation}
\end{theorem}

For the low privacy regime $e^\varepsilon \succeq d$, we select $w=1$ and the $\ell_2$ estimation error is 
\begin{equation}\label{eq:very_low_privacy}
    \E\lVert \hat{   {p}}\lp Y^n\rp -    {p} \rVert^2_2 = O\lp \frac{e^{2\varepsilon}}{n\lp e^\varepsilon-1 \rp ^2} \rp = O\lp\frac{1}{n}\rp. 
\end{equation}
For the regime $e^\varepsilon \prec d$, we select $w = \lb\frac{d}{e^{\varepsilon}+1}\rb$, (see \cite[Proposition III.3]{yebarg2018}) and the $\ell_2$ estimation error is given by
$$  \Theta\lp \frac{d}{n\min \lbp \lp e^\varepsilon -1 \rp^2, e^\varepsilon\rbp} \rp. $$
The estimation error matches the lower bounds in both high and medium privacy regimes \citep{duchi2013local, yebarg2018} and thus is rate-optimal. For the low privacy regime, \cite{yebarg2018} coincides with the $k$-RR scheme from \cite{kairouz16} and achieves optimal rate (i.e. $1/n$) too.  

We will use the rate-optimal distribution estimators $\hat{ {p}}$ to construct an estimator under sparse Bernoulli model with LDP constraint.

\subsubsection*{Estimating $p_i$}
 We use the first half of nodes to estimate $p_i$. For $k\in \lbp 1,..., n/2 \rbp$, node $k$ transmit $Y_k$ according to $Q_{\text{DE}}\lp \cdot | f(   {X}_k) \rp$, and let $\hat{p}_i$ be the rate-optimal estimator of $p_i$ as defined in previous subsection. The $\ell_2$ estimation error of $\hat{   {p}}$ is controlled by 
 \begin{equation}\label{eq:bound_p_i}
 \sum_{i=1}^d \E \lp  \hat{p}_i - p_i\rp^2 \preceq \frac{d}{n \min\lbp e^\varepsilon, (e^\varepsilon - 1)^2\rbp}.
 \end{equation}

We can truncate $\hat{p}_i$ and obtain a better estimator (i.e. with smaller $\ell_2$ risk) since by definition $p_i$ cannot take negative values:
$$\hat{p}^*_i \eqDef \max \lbp  \hat{p}_i, 0 \rbp,$$
and thus $\hat{p}^*_i > 0$ almost surely.

\subsubsection*{Estimating $P_S$}
The second half of nodes are used to estimate $P_S \eqDef \prod_{j=1}^d (1-\theta_j)$, which  maps its observation $X_k$ via
$$  g(   {X}_k) = \begin{cases}
1, \text{ if } \lVert    {X}_k \rVert_1 = 0, \\
0, \text{ else.}
\end{cases} $$
Note that 
\begin{itemize}
	\item$g(   {X}_k) \sim \Ber(P_S)$
	\item $P_S$ is lower bounded by some positive constant (see \cite[Section~4.2]{archayaetal2}):
	\begin{align*}
		P_S 
		& = \prod_{i=1}^d (1-\theta_i) \leq 
		\exp\lp -\sum_{i=1}^d \ln(1-\theta_i) \rp \\
		& = \exp\lp -\sum_{i=1}^d \ln(1-\theta_i) \rp \\
		& =\exp \lp - 1-\sum_{t>1} \lVert    {\theta}\rVert^t_t/t\rp \\
		& \geq\exp \lp - 1-\sum_{t>1} \lVert    {\theta}\rVert^t_2/t\rp \\
		& \geq \frac{1-\frac{1}{\sqrt{2}}}{\exp\lp 1-\frac{1}{\sqrt{2}} \rp}.
	\end{align*}
\end{itemize}

Therefore, for $k\in [n/2+1:n]$, node $k$ transmits $Y_k$ according to $Q_{\text{DE}}\lp \cdot | g(   {X}_k) \rp$, and let $\hat{P}_S$ be the rate-optimal estimator of $P_S$. Estimating $P_S$ is equivalent to distribution estimation problem with $d'=2$ (which may falls into low privacy regime \eqref{eq:very_low_privacy}), and by \eqref{eq:l_2_error} the previous privatization scheme guarantees 

\begin{equation}\label{eq:bound_p_s_1}
\E \lp \hat{P}_S-P_S\rp^2 \preceq 
\frac{e^{2\varepsilon}}{n\lp e^\varepsilon -1\rp^2}.
%\prec \frac{d}{n \max\lbp e^\varepsilon, (e^\varepsilon - 1)^2\rbp}.
\end{equation}
Notice that if $e^\varepsilon \leq d$, then 
$$\frac{e^{2\varepsilon}}{n\lp e^\varepsilon -1\rp^2} \leq \frac{d e^{\varepsilon}}{n\lp e^\varepsilon -1\rp^2} \asymp \frac{d}{n \min\lbp e^\varepsilon, (e^\varepsilon - 1)^2\rbp} $$
(the last "$\asymp$" is derived by separating $e^\varepsilon$ into $e^\varepsilon \succ 1$ and $e^\varepsilon \asymp 1$).
Otherwise $e^\varepsilon \geq 2$ and \eqref{eq:bound_p_s_1} is $O(1/n)$.
Since we already know that $P_S \geq \frac{1-\frac{1}{\sqrt{2}}}{\exp\lp 1-\frac{1}{\sqrt{2}} \rp}$, the truncated estimator 

$$ \hat{P}^*_{S} \eqDef \max \lbp \hat{P}_S, \frac{1-\frac{1}{\sqrt{2}}}{\exp\lp 1-\frac{1}{\sqrt{2}} \rp}\rbp$$
must have smaller $\ell_2$ estimation error and is bounded
\begin{equation}\label{eq:bound_p_s_3}
\frac{1}{\hat{P}^*_S} \leq \frac{\exp\lp 1-\frac{1}{\sqrt{2}} \rp}{1-\frac{1}{\sqrt{2}}} \text{ almost surely}.
\end{equation}

\subsubsection*{Analysis of $\ell_2$ error of $\hat{   {\theta}}$}
Our final estimator is $$ \hat{\theta}_i\lp Y^n \rp \eqDef \frac{ \hat{p}^*_i(Y_1,...,Y_{\frac{n}{2}}) }{\hat{p}^*_i(Y_1,...,Y_{\frac{n}{2}})+ \hat{P}^*_S(Y_{\frac{n}{2}+1},...,Y_n)}. $$ 
The $\ell_2$ error is 
\begin{align*}
	\E\lp \hat{\theta}_i - \theta_i \rp^2 
	& = \E \lp \frac{\hat{p}^*_i}{\hat{p}^*_i+\hat{P}^*_S} - \frac{p_i}{p_i+P_S}\rp^2 \\
	& = \E \lp  \frac{\hat{p}^*_i - {p}_i}{\hat{p}^*_i+\hat{P}^*_S}+ {p}_i \lp \frac{1}{\hat{p}^*_i+\hat{P}^*_S} - \frac{1}{p_i+P_S}\rp\rp^2\\
	& \leq 2\underbrace{\E \lp  \frac{\hat{p}^*_i - p_i}{\hat{p}^*_i+\hat{P}^*_S}\rp^2}_{\text{(a)}} + 2 p_i^2  \underbrace{ \E \lp \frac{1}{\hat{p}^*_i+\hat{P}^*_S} - \frac{1}{p_i+P_S}\rp^2}_{\text{(b)}}
\end{align*}

By \eqref{eq:bound_p_i} and \eqref{eq:bound_p_s_3}, (a) can be bounded by
$$ \text{(a)} \leq 2\E \lb \lp \frac{1}{\hat{p}^*_i+\hat{P}^*_S} \rp^2 \lp \hat{p}^*_i - p_i \rp^2\rb \leq 2\lp \frac{\exp\lp 1-\frac{1}{\sqrt{2}} \rp}{1-\frac{1}{\sqrt{2}}}\rp^2 \E \lp \hat{p}^*_i - p_i \rp^2,$$

By \eqref{eq:bound_p_i}, \eqref{eq:bound_p_s_1} and \eqref{eq:bound_p_s_3} (b) can be bounded by
\begin{align*}
	(b) & =  2 p_i^2 \E \lb \lp \frac{1}{\lp\hat{p}^*_i+\hat{P}^*_S\rp\lp p_i+P_S\rp}\rp^2  \lp \hat{p}^*_i - p_i+\hat{P}^*_S - P_S\rp^2\rb \\
	& \leq 4 \lp \frac{\exp\lp 1-\frac{1}{\sqrt{2}} \rp}{1-\frac{1}{\sqrt{2}}}\rp^4 p_i^2\lp \E\lp\hat{p}^*_i - p_i\rp^2 +\E\lp\hat{P}^*_S - P_S\rp^2 \rp.
\end{align*}

Finally, summing over all $i\in[d]$, we have 
\begin{align*}
	\E \lVert \hat{   {\theta}} -    {\theta} \rVert^2_2 
	&\leq C_0 \sum_i \E \lp \hat{p}^*_i - p_i\rp^2 + C_1 \E \lp \hat{P}^*_S - P_S \rp^2\\ 
	& \asymp \frac{d}{n \min\lbp e^\varepsilon, \lp e^\varepsilon -1\rp^2 \rbp},
\end{align*}
where $C_0$ and $C_1$ are some universal constants.

%% file: sec_s_sparse_high_privacy.tex
The general idea is to group $\lbp \theta_1,...,\theta_d \rbp$ into subgroups $\mcal{G}_1,..,\mcal{G}_s$, such that $\sum_{i\in\mcal{G}_j} \theta_i = O(1)$. If we can do so, then for each sub-group we apply Part~(i) of Corollary~\ref{cor3} with effective sample size $n'_{j} = n\lba \mcal{G}_j \rba/d$, and the resulting $\ell_2$ estimation error will be
$$ \sum_{j=1}^s \frac{\lba \mcal{G}_j \rba}{n'_j\min\{(e^\varepsilon-1)^2,e^\varepsilon,d\}} = \frac{ds}{n \min\{(e^\varepsilon-1)^2,e^\varepsilon,d\}}. $$

In order to grouping $\lbp \theta_1,...,\theta_d \rbp$, in the first phase we estimate each of them up to precision $1/d$ with the first half of samples. This requires roughly $n \approx d^3/\min{\lp \varepsilon^2, 1 \rp}$ samples. Once we obtain a coarse estimate of each $\theta_j$, in the second phase we perform the mechanism for $1$-Sparse Bernoulli model with the rest of the samples and refine the estimate.

\subsubsection*{Phase~1: grouping parameters}
Let $n' = n/2d$, and by assumption 
$$ n' \geq  10\lp\frac{e^\varepsilon+1}{e^\varepsilon -1}\rp^2 d^2 \log d \log n.  $$
We will use $n'$ samples to estimate $\theta_j, \, \forall j \in [d]$. For the $j$-th component of the $i$-sample $X_i(j) \sim \Ber\lp \theta_j \rp$, let $Y_i(j)$ be the $\varepsilon$-privatized version of it, so 
$$ Y_i(j) \sim \Ber\lp \lp\frac{e^\varepsilon+1}{e^\varepsilon -1}\rp \theta_j + \frac{1}{e^\varepsilon+1} \rp. $$
Therefore, by Hoeffding's inequality we have
\begin{align*}
    \Pr \lbp \lba \hat{\theta}_j\lp Y^{n'}(j)\rp - \theta_j \rba \geq \frac{1}{d} \rbp
    & = \Pr \lbp \lba \frac{1}{n}\sum_{i=1}^{n'}\lp\frac{e^\varepsilon+1}{e^\varepsilon -1}\rp\lp Y_i(j) - \E[Y_i(j)] \rp  \rba \geq \frac{1}{d} \rbp \\
    & \leq 2 \exp\lp -\frac{n'}{2d^2 \lp\frac{e^\varepsilon+1}{e^\varepsilon -1}\rp^2 } \rp \\
    & \leq \frac{1}{d^{10} n}.
\end{align*} 
Let $\mcal{E} = \lbp \exists j, \,\, \lba \hat{\theta}_j - \theta_j \rba \geq 1/d \rbp$ be the event of failure. By union bound, 
$$ \Pr \lbp \mcal{E} \rbp \leq \frac{1}{n d^9}. $$
For the rest of analysis, we will condition on $\mcal{E}^c$. 

Since $\forall j, \, \hat{\theta}_j \leq \theta_j + 1/d$, we must have $\sum_{j}\hat{\theta}_j \leq s+1$, and therefore we can find $s$ groups $\mcal{G}_1,..,\mcal{G}_s$, such that 
$$ \forall k \in [s], \, \sum_{j \in \mcal{G}_k}\hat{\theta}_j \leq 2. $$
On the other hand, $\forall j, \, \theta_j \leq \hat{\theta}_j + 1/d$, so we also have
$$ \forall k \in [s], \, \sum_{j \in \mcal{G}_k}\theta_j \leq 3. $$

\subsubsection*{Phase~2: reducing to $1$-Bernoulli model}
Conditioning on $\mcal{E}^c$ and applying Part~(i) of Corollary~\ref{cor3} for each group $\mcal{G}_j$, with effective sample size $n'_{j} = n\lba \mcal{G}_j \rba/d$, the $\ell_2$ estimation error is upper bounded by
$$ \frac{\lba \mcal{G}_j \rba}{n'_j\min\{(e^\varepsilon-1)^2,e^\varepsilon,d\}} =  \frac{d}{n \min\{(e^\varepsilon-1)^2,e^\varepsilon,d\}}. $$
Summing over $s$ groups yields
\begin{equation}\label{eq:upp_bdd_1}
    \E\lb \lVert \hat{{\theta}} - {\theta} \rVert^2_2 \Big| \mcal{E}^c\rb \asymp \frac{ds}{n \min\{(e^\varepsilon-1)^2,e^\varepsilon,d\}}.
\end{equation}

On the other hand, if phase~1 fails, we have a trivial upper bound:
$$ \lVert \hat{{\theta}} - {\theta} \rVert^2_2 \leq d. $$
Therefore 
\begin{align*}
    \E\lb \lVert \hat{{\theta}} - {\theta} \rVert^2_2\rb 
    & \leq \Pr\lbp \mcal{E} \rbp\E\lb \lVert \hat{{\theta}} - {\theta} \rVert^2_2 \Big| \mcal{E}\rb+\E\lb \lVert \hat{{\theta}} - {\theta} \rVert^2_2 \Big| \mcal{E}^c\rb \\
    & \leq \frac{2}{n d^8} + \E\lb \lVert \hat{{\theta}} - {\theta} \rVert^2_2 \Big| \mcal{E}^c\rb \\
    & \asymp \frac{ds}{n \min\{(e^\varepsilon-1)^2,e^\varepsilon,d\}},
\end{align*}
achieving the desired result.

%% file: sec_general_sparsity.tex
Let us first derive the lower bound for estimation error. Restricting our attention to $\frac{s}{2d}\leq\theta_j\leq\frac{s}{d}$, the score for each component of the Bernoulli model \eqref{eq:bern_score} is sub-Gaussian with parameter 
$$ \sigma^2 \leq \frac{c_0d^2}{s^2\log\lp\frac{d}{s}\rp}. $$
This can be checked by letting
\begin{align*}
\sigma = \max\left\{\frac{1}{\theta_j\sqrt{\log\frac{1}{\theta_j}}} \; , \; \frac{1}{(1-\theta_j)\sqrt{\log\frac{1}{(1-\theta_j)}}}  \right\} \; ,
\end{align*}
and then
\begin{align*}
\mathbb{E}\left[e^{\left(\frac{S_{\theta_i}(X_j)}{\sigma}\right)^2}\right] & = \theta_j e^{\left(\frac{1}{\theta_j\sigma}\right)^2} + (1-\theta_j) e^{\left(\frac{1}{(1-\theta_j)\sigma}\right)^2} \\
& \leq 2
\end{align*}
and thus $S_{\theta}(X)$ is $\sigma^2$ sub-Gaussian. Then note that for $\theta_j \leq \frac{1}{2}$ the first term is the maximizer and
$$\sigma^2 = \frac{1}{\theta_j^2\log\frac{1}{\theta_j}} \leq \frac{c_0d^2}{s^2\log\lp\frac{d}{s}\rp}\; .$$
Applying Propositon \ref{prop4}, we obtain the desired lower bound. In the rest of the section, we give an explicit construction of $Q$ and $\hat{  {\theta}}$ and characterize the error for this mechanism.

\subsubsection*{$k$-Randomized Response ($k$-RR) Scheme }
If the support size of the input alphabet is $k$, $k$-RR scheme outputs the input symbol with probability $e^\varepsilon/(k-1+e^\varepsilon)$ and the rest of $k-1$ symbols with probability $1/(k-1+e^\varepsilon)$:
$$
Q(y|x) = 
\begin{cases}
\frac{e^\varepsilon}{(k-1)+e^\varepsilon}, \text{ if } y=x,\\
\frac{1}{(k-1)+e^\varepsilon}, \text{ if } y\neq x.
\end{cases}
$$
$k$-RR scheme works well for low privacy regime, i.e. when $e^\varepsilon$ is large, and will be used later as our privatization mapping. 

In general, as long as $e^\varepsilon \approx k$, the privatization error is with the same order of estimation error, so we can estimate the discrete distribution without increasing additional estimation error by too much.

\subsubsection*{LDP Scheme via Sub-sampling and $k$-RR}

In our problem, for each node we aim to transmit its local observation $  {X}_k$ reliably to the fusion center, and with high probability there will be roughly $s$'s $1$ in $  {X}_k$, so the expected number of possible $  {X}_k$ is roughly ${d \choose s}$. Unfortunately this means we need ${d \choose s} \approx \exp \lp s \log d \rp$ symbols to represent it, and notice that $\varepsilon \leq s\log d$, so we cannot send $  {X}_k$ reliably under $\varepsilon$-local privacy constraint.

To address this issue, we use \textit{sub-sampling} trick to reduce the effective support size. First let $k$ be the largest integer such that
$$ \sum_{0\leq i \leq k} {d \choose i} \leq \exp\lp \varepsilon-10\log d\rp. $$
Notice that $k \asymp \frac{\varepsilon}{\log d}$ since
\begin{align*}
	\sum_{0\leq i \leq k} {d \choose i} \leq d {d \choose k} \leq \exp \lp k\log d + \log d\rp,
\end{align*}
so we must have 
$$ k \geq \frac{\varepsilon}{\log d} -11 \succeq \frac{\varepsilon}{\log d}. $$

Next, for each local observation $  {X}_k$, consider the sub-sampled version $\tilde{  {X}}_k$ as follows:

$$ \tilde{  {X}}_k \eqDef 
\begin{cases}
{X}_k, \text{ if } \lVert   {X} \rVert_1 \leq k,\\
\text{ randomly keep } k \text{'s }1 \text{ in }   {X}, \text{ if } \lVert   {X}_k \rVert_1 > k.
\end{cases}$$
If we let $R_k$ be the reciprocal of sampling rate $\frac{\max\lp \lVert   {X} \rVert_1, k \rp}{k}$, then 
$$ \E\lb \E\lb R_k \cdot \tilde{  {X}}_k(i) \Big| R_k  \rb\rb = \E\lb \E\lb   {X}_k(i) \Big| R_k \rb\rb = \theta_i. $$
Finally, each node transmits $ \tilde{  {X}}_k$ via $k$-RR scheme with privacy level $ \varepsilon$:

$$ Q\lp   {Y} | \tilde{  {X}} \rp = 
\begin{cases}
\frac{e^{\varepsilon}}{(N-1)+e^{\varepsilon}}, \text{ if }   {Y} = \tilde{  {X}},\\
\frac{1}{(N-1)+e^{\varepsilon}}, \text{ if }  {Y} \neq \tilde{X},
\end{cases}
$$
where $N$ is the number of possible $\tilde{ {X}}_k$:
$$ N \eqDef \sum_{i \leq k} {d \choose i} \leq \exp\lp \varepsilon \rp/d^{10}. $$
We also use $p_e$ to denote the probability of privatization error, i.e.

$$ p_e \eqDef \Pr\lbp Q\lp   {Y} | \tilde{  {X}} \rp \neq \tilde{  {X}} \rbp = \frac{N-1}{(N-1)+e^{\varepsilon}} \leq \frac{1}{d^{10}}. $$

Now we compute $\Pr\lbp   {Y}_k(i) = 1 | R_k\rbp$ for some node $k$:
\begin{align*}
	\Pr\lbp   {Y}_k(i) = 1 |R_k \rbp  = &\Pr\lbp \tilde{  {X}}_k(i) = 1 \cap   {Y}_k = \tilde{  {X}}_k |R_k\rbp \\
	+&\Pr\lbp   {Y}_k(i) = 1 \cap \tilde{  {X}}_k(i) = 1 \cap   {Y}_k \neq \tilde{  {X}}_k | R_k\rbp \\
	+&\Pr\lbp   {Y}_k(i) = 1 \cap \tilde{  {X}}_k(i) \neq 1 \cap   {Y}_k \neq \tilde{  {X}}_k | R_k\rbp \\
\end{align*}
It not hard to see that the first term is $ (1-p_e) \cdot \E\lb \tilde{  {X}}_k(i)|R_k \rb$, and we further derive the second and the third terms:
\begin{align*}
	&\Pr\lbp   {Y}_k(i) = 1 \cap \tilde{  {X}}_k(i) = 1 \cap   {Y}_k \neq \tilde{  {X}}_k | R_k\rbp \\
	& = p_e\cdot\E\lb \tilde{  {X}}_k(i)|R_k \rb\cdot\Pr\lbp   {Y}_k(i) = 1 | \tilde{  {X}}_k(i) = 1,   {Y}_k \neq \tilde{  {X}}_k, R_k\rbp \\
	& = p_e\cdot\E\lb \tilde{  {X}}_k(i)|R_k\rb \frac{\sum_{1<i\leq k} {d-1 \choose i-1}}{\sum_{0 \leq i\leq k} {d \choose i}-1}, \\
\end{align*}
and 
\begin{align*}
	&\Pr\lbp   {Y}_k(i) = 1 \cap \tilde{  {X}}_k(i) \neq 1 \cap   {Y}_k \neq \tilde{  {X}}_k | R_k\rbp \\
	& = p_e\cdot\lp 1-\E\lb \tilde{  {X}}_k(i)|R_k \rb\rp \cdot\Pr\lbp   {Y}_k(i) = 1 | \tilde{  {X}}_k(i) \neq 1,   {Y}_k \neq \tilde{  {X}}_k, R_k\rbp \\
	& = p_e\cdot\lp 1-\E\lb \tilde{  {X}}_k(i)|R_k\rb\rp \frac{\sum_{1\leq i\leq k} {d-1 \choose i}}{\sum_{0 \leq i\leq k} {d \choose i}-1}. \\
\end{align*}
Summing the three terms together, we have 
\begin{equation}\label{eq:def_A_B}
	\Pr\lbp   {Y}_k(i) = 1 |R_k \rbp = A\cdot\E\lb \tilde{  {X}}_k(i)|R_k\rb + B, 
\end{equation}
for some known constants $A$ and $B$.
Notice that $1 \geq A \geq (1-2p_e) \geq (1-\frac{2}{d^{10}})$, and $B \leq p_e \leq \frac{1}{d^{10}}$, which implies $\E\lb   {Y}_k(i) |R_k \rb \approx \E\lb \tilde{  {X}}_k(i)|R_k\rb$.

If we know $R_k$, 
then $ \hat{\theta_i} \eqDef R_k\cdot\frac{  {Y}_k-B}{A}$ is an unbiased estimator of $\theta_i$ since
\begin{align*}
	\E\hat{\theta_i} 
	& = \E \lb R_k\cdot\frac{  {Y}_k(i)-B}{A}\rb = \E \lb R_k \E \lb\frac{  {Y}_k(i)-B}{A} \Big| R_k \rb\rb\\
	& =\E \lb \E \lb   {X}_k(i) \Big| R_k \rb\rb = \E   {X}_k(i) = \theta_i. 
\end{align*} 

Unfortunately, we cannot directly obtain $R_k$ since otherwise the LDP constraint will be violated. So instead, we replace $R_k$ with an estimate of $\E\lb R_k \rb$, denoted as $\hat{\mu_{R}}$, in our estimator:
$$ \hat{\theta_i}({Y}_k) \eqDef \hat{\mu}_R\cdot\frac{ {Y}_k(i)-B}{A}.$$
Notice that 
\begin{align*}
	\E\lb R_i \rb = \E\lb\frac{\max\lp \lVert X_i \rVert_1, k \rp}{k}\rb 
	& = \E\lb \frac{\lVert X_i\rVert_1}{k} \rb + \E\lb \frac{\max\lp k - \lVert X_i \rVert_1, 0 \rp}{k} \rb \\ 
	& = \frac{s}{k}-\E\lb \frac{\max\lp k-\lVert X_i \rVert_1, 0 \rp}{k} \rb. 
\end{align*}
Since $\E\lb \frac{\max\lp k - \lVert X_i \rVert_1, 0 \rp}{k} \rb$ is bounded by $1$ and $\varepsilon = \Omega(1)$, by using first $\theta(s) = o(n)$ of samples, we can estimate $\E[R_i]$ to precision $1/s$ privately, i.e. $\E\lb \lp \hat{\mu}_R - \E\lb R_i \rb \rp^2 \rb \leq \frac{1}{s}$.

%Notice that  $\hat{\theta_i}$ is still an unbiased estimator because we only add a zero-mean independent noise on $R_k$. Since now both $  {Y}_k$ and $\tilde{R}_k$ are transmitted through two $\frac{\varepsilon}{2}$-LDP schemes, by the composition theorem of differential privacy, the overall scheme is $\varepsilon$-LDP.

Our final estimator is the aggregation of $  {Y}_1,...,  {Y}_n$:
$$ \hat{  {\theta}}\lp   {Y}_1,...,   {Y}_n \rp \eqDef \frac{1}{n}\sum_{k=1}^n \hat{\mu}_R \frac{  {Y}_k - B}{A}, $$
where $A, B$ are constants defined in \eqref{eq:def_A_B} and are independent of $  {\theta}$. It remains to show that for this mechanism and estimator,
$$\mathbb{E}\|\hat\theta(Y_1,\ldots,Y_n)-\theta\|_2^2 \preceq \frac{s^2 \log d}{n\varepsilon} \; .$$

%% file: sec_general_sparsity_app.tex
\subsubsection*{Analysis of $\ell_2$ error}
Now let us analyze the $\ell_2$ error of $\hat{  {\theta}}$. As stated in previous section, $\hat{\theta}_i$ is unbiased to $\theta_i$, so

\begin{align*}
    &\E\lb\lp \frac{1}{n}\sum_{k=1}^n \hat{\mu}_R\frac{  {Y}_k(i) - B}{A}-\theta_i \rp^2\rb \\
    & =  \frac{1}{n^2}\sum_{k=1}^n \E\lb\lp \hat{\mu}_R\frac{  {Y}_k(i) - B}{A}-\theta_i \rp^2\rb \\
    & =  \frac{1}{n}\E\lb\lp \hat{\mu}_R\frac{  {Y}_1(i) - B}{A}-\theta_i \rp^2\rb \\
    &\leq  \frac{3}{n}  \underbrace{\E\lb\lp \hat{\mu}_R\frac{  {Y}_1(i) - B}{A}- \hat{\mu}_R\tilde{  {X}}_1(i)\rp^2\rb}_{\text{(a)}}+
   \frac{3}{n} \underbrace{\E\lb\lp \hat{\mu}_R\tilde{  {X}}_1(i) - R_1\tilde{  {X}}_1(i)\rp^2\rb}_{\text{(b)}}\\
     &  \hphantom{\leq  \frac{3}{n} \underbrace{\E\lb\lp \hat{\mu}_R\frac{  {Y}_1(i) - B}{A}- \hat{\mu}_R\tilde{  {X}}_1(i)\rp^2\rb}_{\text{(a)}}+++}
     +\frac{3}{n}\underbrace{\E\lb\lp R_1\tilde{  {X}}_1(i) - \theta_i\rp^2\rb}_{\text{(c)}}
    .
\end{align*}
Note that (a) and (b) can be viewed as \textit{privatization} errors due to the LDP constraint, and (c) is the estimation error. We bound (a), (b) and (c) separately. 

To bound (a), we leverage the following facts
\begin{itemize}
    \item $ 1 > A \geq (1-\frac{2}{d^{10}})$,
    \item $ 0 < B < \frac{1}{d^{10}}$,
    \item $\Pr\lbp \tilde{  {X}}_1(i)=   {Y}_1(i)\rbp \geq 1- p_e \geq 1- \frac{1}{d^{10}}$.
\end{itemize}
 
\begin{align*}
    &\E\lb\lp \hat{\mu}_R\frac{  {Y}_1(i) - B}{A}- \hat{\mu}_R\tilde{  {X}}_1(i)\rp^2\rb\\
    & = \E\lb\lp\hat{\mu}_R\rp^2\rb \cdot \E\lb\E \lb\lp \frac{  {Y}_1(i) - B}{A}- \tilde{  {X}}_1(i)\rp^2 \Big| R_1 \rb\rb \\
    & \leq \lp\hat{\mu}_R\rp^2
    \Pr\lbp \tilde{  {X}}_1(i)=   {Y}_1(i)\rbp\cdot\lp \lp \frac{1-B-A}{A}\rp^2+ \lp\frac{B}{A}\rp^2\rp \\
    &+ \E\lb\lp\hat{\mu}_R\rp^2\rb
    \Pr\lbp \tilde{  {X}}_1(i)\neq   {Y}_1(i)\rbp\cdot\lp \lp \frac{A-B}{A}\rp^2+ \lp\frac{A+B}{A}\rp^2\rp\\
    & \leq \E\lb\lp\hat{\mu}_R\rp^2\rb\lp\lp \lp \frac{1-B-A}{A}\rp^2+ \lp\frac{B}{A}\rp^2\rp+ p_e \lp \lp \frac{A-B}{A}\rp^2+ \lp\frac{A+B}{A}\rp^2\rp\rp \\
    & \leq \frac{C_0}{d^{10}} \E\lb\lp\hat{\mu}_R\rp^2\rb. \\
\end{align*}
Note that, $\E\lb\lp\hat{\mu}_R\rp^2\rb \leq \E \lb R_1^2\rb + \Var\lp \hat{\mu}_R \rp  \leq 2d^2$, so term (a) can eventually be bounded by 
$$ \text{(a)} \leq \frac{C_0}{d^{8}}, $$ %\leq \frac{C_0}{d^9 \lp \log d\rp^2}, $$
for some constant $C_0$.
To bound (b), observe that 
\begin{align*}
    \sum_{i=1}^d\E\lb\lp \hat{\mu}_R\tilde{  {X}}_1(i) - R_1\tilde{  {X}}_1(i)\rp^2\rb 
    & = \sum_{i=1}^d\E\lb\tilde{  {X}}_1(i)^2\cdot\lp \hat{\mu}_R - R_1\rp^2\rb\\
    & = \E\lb  \lp \sum_{i=1}^d\tilde{{X}}_1(i)\rp \cdot\lp \hat{\mu}_R - R_1\rp^2\rb\\
    & = k \cdot \E\lb \lp \hat{\mu}_R - R_1\rp^2\rb\\
    & \leq 2k \cdot \lp \E\lb \lp \hat{\mu}_R - \E\lb R_1 \rb\rp^2 \rb + \frac{1}{k^2}\Var\lp \max{\lp k, \left\| X_1 \right\|_1\rp}\rp\rp\\
    & \overset{\text{(1)}}{\leq} \frac{2k}{s}+\frac{2}{k} \cdot \Var\lp \left\| X_1 \right\|_1\rp\\
    & \overset{\text{(2)}}{\leq}  C_1 s/k,
\end{align*}
where \text{(1)} is due to  $\hat{\mu}_R$ is of precision $O(1/s)$ by using first $o(n)$ samples, and \text{(2)} is because $k \asymp \varepsilon/\log d \preceq s$.
Finally we bound (c) as follow:
\begin{align*}
    \E\lb\lp R_1\tilde{  {X}}_1(i) - \theta_i\rp^2\rb
    & = \E\lb R_1^2\tilde{  {X}}^2_1(i)\rb  - \theta_i^2 \\
    & \leq \E\lb\E \lb R_1^2\tilde{  {X}}_1(i) \Big| R_1 \rb\rb \\
    & = \E\lb\E \lb R_1   {X}_1(i) \Big| R_1 \rb\rb \\
    & \leq \E\lb \lp \frac{\lVert   {X}_1\rVert_1}{k}+1\rp   {X}_1(i) \rb \\
    & \leq \E\lb \frac{\lVert   {X}_1\rVert_1}{k}  {X}_1(i) \rb+ \theta_i.
\end{align*}

Combining (a), (b) and (c) and summing across all dimensions $i=1,...,d$, we obtain
\begin{align*}
    \E\lVert \hat\theta -  \theta\rVert_2^2
    & = \sum_{i=1}^d \E\lb \lp \hat{\theta}_i - \theta_i \rp^2 \rb \\
    & \leq \frac{1}{n}\lp C_0'\frac{1}{\lp d^7 \rp^2} + C_1' \frac{s}{k} + C_2' \lp \frac{\E \lb \lVert{X}_1\rVert^2_1\rb}{k} +s\rp \rp \\
    & \asymp \frac{s^2\log d}{n \varepsilon},
\end{align*}
where in the last step we bound the second moment of Poisson binomial distribution by 
$$ \E\lVert   {X}_1 \rVert^2_1 = \lp \sum_i \theta_i\rp^2 +\sum_i \theta_i(1-\theta_i) \leq s^2 + s,$$
and observe that $k$ is $\Omega\lp \frac{\varepsilon}{\log d}\rp$.

%% file: main_511.bbl
% Generated by IEEEtran.bst, version: 1.14 (2015/08/26)
\newcommand{\noopsort}[1]{}
\begin{thebibliography}{10}
\providecommand{\url}[1]{#1}
\csname url@samestyle\endcsname
\providecommand{\newblock}{\relax}
\providecommand{\bibinfo}[2]{#2}
\providecommand{\BIBentrySTDinterwordspacing}{\spaceskip=0pt\relax}
\providecommand{\BIBentryALTinterwordstretchfactor}{4}
\providecommand{\BIBentryALTinterwordspacing}{\spaceskip=\fontdimen2\font plus
\BIBentryALTinterwordstretchfactor\fontdimen3\font minus
  \fontdimen4\font\relax}
\providecommand{\BIBforeignlanguage}[2]{{%
\expandafter\ifx\csname l@#1\endcsname\relax
\typeout{** WARNING: IEEEtran.bst: No hyphenation pattern has been}%
\typeout{** loaded for the language `#1'. Using the pattern for}%
\typeout{** the default language instead.}%
\else
\language=\csname l@#1\endcsname
\fi
#2}}
\providecommand{\BIBdecl}{\relax}
\BIBdecl

\bibitem{warner}
S.~L. Warner, ``Randomized response: A survey technique for eliminating evasive
  answer bias,'' \emph{Journal of the American Statistical Association},
  vol.~60, no. 309, pp. 63--69, 1965.

\bibitem{dwork1}
C.~Dwork, F.~McSherry, K.~Nissim, and A.~Smith, ``Calibrating noise to
  sensitivity in private data analysis,'' in \emph{Theory of Cryptography
  Conference}, S.~Halevi and T.~Rabin, Eds.\hskip 1em plus 0.5em minus
  0.4em\relax Springer, Berlin, Heidelberg, 2006.

\bibitem{dwork2}
C.~Dwork and J.~Lei, ``Differential privacy and robust statistics,'' in
  \emph{Proceedings of the 41st annual ACM symposium on Theory of
  computing}.\hskip 1em plus 0.5em minus 0.4em\relax ACM, 2009, pp. 371--380.

\bibitem{dwork3}
C.~Dwork, ``Differential privacy,'' \emph{Automata}, pp. 1--12, 2016.

\bibitem{Cramer}
H.~Cram\'er, \emph{Mathematical Methods of Statistics}.\hskip 1em plus 0.5em
  minus 0.4em\relax Princeton Univ. Press, 1946.

\bibitem{Rao}
C.~R. Rao, ``Information and the accuracy attainable in the estimation of
  statistical parameters,'' \emph{Bulletin of the Calcutta Mathematical
  Society}, vol.~37, 1945.

\bibitem{Cover--Thomas2006}
T.~M. Cover and J.~A. Thomas, \emph{Elements of Information Theory},
  2nd~ed.\hskip 1em plus 0.5em minus 0.4em\relax New York: Wiley, 2006.

\bibitem{Tsybakov2008}
A.~Tsybakov, \emph{Introduction to Nonparametric Estimation}.\hskip 1em plus
  0.5em minus 0.4em\relax Springer-Verlag, 2008.

\bibitem{Hajek1972local}
J.~H{\'a}jek, ``Local asymptotic minimax and admissibility in estimation,'' in
  \emph{Proceedings of the sixth Berkeley symposium on mathematical statistics
  and probability}, vol.~1, 1972, pp. 175--194.

\bibitem{Vandervaart2000}
A.~W. Van~der Vaart, \emph{Asymptotic statistics}.\hskip 1em plus 0.5em minus
  0.4em\relax Cambridge university press, 2000, vol.~3.

\bibitem{gill}
R.~D. Gill and B.~Y. Levit, ``Applications of the van trees inequality: a
  {B}ayesian cram\'er-rao bound,'' \emph{Bernoulli}, vol.~1, no. 1/2, pp.
  059--079, 1995.

\bibitem{duchi2013local}
J.~C. Duchi, M.~I. Jordan, and M.~J. Wainwright, ``Local privacy and
  statistical minimax rates,'' in \emph{2013 IEEE 54th Annual Symposium on
  Foundations of Computer Science}.\hskip 1em plus 0.5em minus 0.4em\relax
  IEEE, 2013, pp. 429--438.

\bibitem{yebarg2018}
M.~Ye and A.~Barg, ``Optimal schemes for discrete distribution estimation under
  locally differential privacy,'' \emph{IEEE Transactions on Information
  Theory}, vol.~64, no.~8, 2018.

\bibitem{duchi2019lower}
J.~Duchi and R.~Rogers, ``Lower bounds for locally private estimation via
  communication complexity,'' in \emph{Proceedings of the 32nd Conference On
  Learning Theory}, vol.~99.\hskip 1em plus 0.5em minus 0.4em\relax PMLR, 2019,
  pp. 1161--1191.

\bibitem{allerton}
L.~P. Barnes, Y.~Han, and A.~\"Ozg\"ur, ``A geometric characterization of
  fisher information from quantized samples with applications to distributed
  statistical estimation,'' in \emph{Proceedings of the 56th Annual Allerton
  Conference on Communication, Control, and Computing}, 2018.

\bibitem{isit}
------, ``Fisher information for distributed estimation under a blackboard
  communication protocol,'' \emph{Proceedings of the IEEE International
  Symposium on Information Theory (ISIT), Paris, France}, 2019.

\bibitem{barnes}
------, ``Lower bounds for learning distributions under communication
  constraints via fisher information,'' \emph{\emph{arXiv preprint},
  \emph{arXiv}:1902.02890}, 2019.

\bibitem{archayaetal2}
J.~Acharya, C.~Canonne, and H.~Tyagi, ``Distributed simulation and distributed
  inference,'' \emph{\emph{arXiv preprint}, \emph{arXiv}:1804.06952}, 2018.

\bibitem{archayaetal}
J.~Acharya, C.~L. Canonne, and H.~Tyagi, ``Inference under information
  constraints: Lower bounds from chi-square contraction,'' in \emph{Proceedings
  of the 32nd Conference on Learning Theory}, vol.~99.\hskip 1em plus 0.5em
  minus 0.4em\relax Phoenix, USA: PMLR, 25--28 Jun 2019, pp. 3--17.

\bibitem{ruan}
F.~Ruan and J.~C. Duchi, ``The right complexity measure in locally private
  estimation: It is not the fisher information,'' \emph{\emph{arXiv preprint},
  \emph{arXiv}:1806.05756}, 2018.

\bibitem{rohde}
A.~Rohde and L.~Steinberger, ``Geometrizing rates of convergence under local
  differential privacy constraints,'' \emph{\emph{arXiv preprint},
  \emph{arXiv}:1805.01422}, 2018.

\bibitem{borovkov}
A.~A. Borovkov, \emph{Mathematical Statistics}.\hskip 1em plus 0.5em minus
  0.4em\relax Gordon and Breach Science Publishers, 1998.

\bibitem{versh}
R.~Vershynin, ``Introduction to the non-asymptotic analysis of random
  matrices,'' \emph{arXiv preprint \emph{arXiv}:1011.3027}, 2010.

\bibitem{kairouz16}
P.~Kairouz, K.~Bonawitz, and D.~Ramage, ``Discrete distribution estimation
  under local privacy,'' in \emph{Proceedings of The 33rd International
  Conference on Machine Learning}, vol.~48, New York, New York, USA, 20--22 Jun
  2016, pp. 2436--2444.

\end{thebibliography}
